\documentclass[twoside,11pt]{article}

\usepackage[preprint]{jmlr2e}

\usepackage[T1]{fontenc}
\usepackage{lmodern}

\usepackage{amsmath,mathtools}

\usepackage{subcaption}
\usepackage{booktabs}

\usepackage[capitalise,noabbrev]{cleveref}

\usepackage{enumitem}
\usepackage{rotating}
\usepackage{placeins}
\usepackage{lastpage}

\newenvironment{interpretation}{\par\noindent{\bf Interpretation.}\itshape}{\par}

\newcommand{\X}{\mathcal{X}}

\newcommand{\dmu}{\delta_\mu}
\newcommand{\kmu}{k_\mu}

\newcommand{\sigmasq}{\sigma^2}

\newcommand{\E}{\mathbb{E}}

\newcommand{\Cov}{\operatorname{Cov}}

\newcommand{\rank}{\operatorname{rank}}

\newcommand{\argmin}{\operatorname*{argmin}}

\newcommand{\R}{\mathbb{R}}

\newcommand{\JS}{\mathrm{JS}}
\newcommand{\KL}{\mathrm{KL}}
\newcommand{\CE}{\mathrm{CE}}
\newcommand{\Hsq}{H^2}
\newcommand{\thmref}[1]{Thm.~\ref{#1}}
\newcommand{\propref}[1]{Prop.~\ref{#1}}
\newcommand{\corref}[1]{Cor.~\ref{#1}}
\newcommand{\defref}[1]{Def.~\ref{#1}}

\begin{document}

\jmlrheading{X}{2026}{1-\pageref{LastPage}}{X/26}{X/26}{dalla-riva26a}{Giulio V.~Dalla Riva}
\ShortHeadings{Task Ecologies and World-Tracking Representations in LLMs}{Dalla Riva}
\firstpageno{1}

\title{Task Ecologies and the Evolution of\\
World-Tracking Representations in Large Language Models}

\author{\name Giulio Valentino Dalla Riva \email me@gvdallariva.net \\
        \addr Baffelan O\"U \\
        \addr \url{https://www.baffelan.com}}

\editor{}

\maketitle

\begin{abstract}
We study language models as evolving model organisms and ask
when autoregressive next-token learning selects for world-tracking
representations. For any encoding of latent world states, the
Bayes-optimal next-token cross-entropy decomposes into the irreducible
conditional entropy plus a Jensen--Shannon excess term. That excess
vanishes if and only if the encoding preserves the training ecology's
equivalence classes. This yields a precise notion of ecological
veridicality for language models and identifies the minimum-complexity
zero-excess solution as the quotient partition by training equivalence.
We then determine when this fixed-encoding analysis applies to
transformer families: frozen dense and frozen Mixture-of-Experts
transformers satisfy it, in-context learning does not enlarge the
model's separation set, and per-task adaptation breaks the premise. The
framework predicts two characteristic failure modes: simplicity pressure
preferentially removes low-gain distinctions, and training-optimal
models can still incur positive excess on deployment ecologies that
refine the training ecology. A conditional dynamic extension shows how
inter-model selection and post-training can recover such gap
distinctions under explicit heredity, variation, and selection
assumptions. Exact finite-ecology checks and controlled microgpt
experiments validate the static decomposition, split-merge threshold,
off-ecology failure pattern, and two-ecology rescue mechanism in a
regime where the relevant quantities are directly observable. The goal
is not to model frontier systems at scale, but to use small language
models as laboratory organisms for theory about representational
selection.
\end{abstract}


\section{Introduction}\label{sec:intro}

Recent work on language-model representations asks whether optimization drives models toward internal structure that tracks the world, or only toward whatever distinctions are locally useful for prediction. The Platonic Representation Hypothesis \citep{huh_etal_2024} argues that task generality, capacity, and simplicity jointly push learned representations toward a shared statistical model of reality; \citet{groeger_etal_2026} challenge the strongest form of that claim, showing that much apparent global alignment is a scale confound and that the robust signal is local-neighborhood rather than global-spectral convergence. Debates about whether language models develop ``world models'' or ``understanding'' \citep{bender_koller_2020,aguerayarcas_2022,mitchell_krakauer_2023,vandijk_etal_2023,cuskley_etal_2024,loru_etal_2025} and empirical demonstrations of domain-specific internal structure \citep{li_etal_2023,gurnee_tegmark_2024,nanda_etal_2023,taniguchi_etal_2025} concern the same issue. We isolate two parts of it that we can state exactly: for a fixed training ecology, which latent distinctions must an autoregressive model preserve in order to achieve Bayes-optimal next-token loss? And under explicit heredity, variation, and selection assumptions on model lineages, what population-level pressure does inter-model competition exert on those distinctions?

Throughout, ``representation'' means an encoding of latent world states into behavioural distinctions: which states the model keeps apart, which it merges, and which differences survive into its next-token predictions. This is close in spirit to the ecological-veridicality framework developed in evolutionary perception, where \citet{hoffman_etal_2015} showed that single-task selection generically favors non-veridical encodings, \citet{berke_etal_2022} showed by simulation that multi-task selection reverses this, and \citet{dallariva_2026} provided the full theory: the separation structure of the task ecology determines which distinctions are preserved, and population-level convergence requires explicit mutation-selection assumptions. An encoding is \emph{ecologically veridical} when it may merge task-equivalent states but not ecology-separated ones. We carry that logic into frozen autoregressive transformers.

Several nearby literatures frame parts of this problem. In multi-task representation learning, \citet{baxter_2000} and \citet{maurer_etal_2016} show that shared representations improve sample complexity, but they do not characterize the exact representational object selected by the autoregressive loss. \citet{lobashev_2025} gives a Bayesian route to convergence in the large-data limit, but attributes failure mainly to capacity mismatch. On the neural-theory side, \citet{wang_johnston_fusi_2025} prove approximately orthogonal latent-variable representations for feedforward networks at global minima, while mechanistic interpretability supplies architectural analogues rather than ecological theorems: \citet{elhage_etal_2021_framework} formalize the transformer residual stream as a shared communication channel, \citet{elhage_etal_2022_superposition} give a capacity-pressure account of feature storage, and \citet{gurnee_etal_2025_linebreaks} show that next-token training can induce low-dimensional internal geometry for structural task variables. The information-bottleneck literature \citep{tishby_etal_1999} is also adjacent, but our object is more concrete: the minimum-complexity encoding that achieves zero excess next-token loss under a fixed ecology.

We study that question in small transformers used as model organisms: systems simple enough that we can inspect induced partitions, exact finite-ecology quantities, and population-level selection trajectories directly. This is a methodological use of model organisms in the sense discussed by \citet{hubinger_etal_2024} and \citet{paez_2024}; \cref{sec:lab} makes the laboratory regime concrete.

We make four main contributions. First, we prove that the Bayes-optimal next-token loss induced by an encoding decomposes exactly into an irreducible entropy term plus a Jensen--Shannon excess term, and that this excess vanishes exactly when the encoding preserves the task-equivalence classes of the training ecology. This is a theorem about the target of the Bayes-optimal next-token objective under a fixed ecology, not a convergence theorem for realistic SGD. Second, we identify when that pressure is well-defined for transformer architectures: frozen dense and frozen Mixture-of-Experts transformers satisfy the required fixed-encoding conditions, in-context learning does not enlarge the separation set, and per-task adaptation changes the encoding. Third, we characterize the simplest zero-excess solution: the minimum-complexity encoding is exactly the quotient partition $W/{\sim_\mu}$, which preserves all and only the distinctions the ecology supports. Fourth, we add a conditional dynamic extension: under explicit heredity, variation, and selection assumptions on model lineages, inter-model selection pushes toward lower ecological excess loss, and a two-ecology mechanism shows how post-training can rescue distinctions weakly supported by the token ecology alone.

We proceed as follows. \Cref{sec:lab} introduces the laboratory model organism. \Cref{sec:ecology,sec:frozen,sec:static} formalise the autoregressive task ecology and establish the static optimality results. \Cref{sec:evolution} states when the ecological-veridicality population dynamics can be imported to model lineages and develops the two-ecology extension. \Cref{sec:mincomplex} develops the minimum-complexity and simplicity-pressure results. \Cref{sec:conclusion} collects the failure predictions, production-scale implications, geometric limits, and concluding discussion, with supplementary results in the appendices.

\section{The Laboratory LLM Organism}\label{sec:lab}

We build our empirical results on a single \emph{model organism}: a small Julia implementation of a frozen autoregressive transformer inspired by the architectural template of \citeauthor{karpathy_2026}'s~(\citeyear{karpathy_2026}) \emph{microgpt}. As in the model-organisms methodology discussed by \citet{hubinger_etal_2024} and \citet{paez_2024}, its value lies not in scale or ecological realism, but in the fact that we can directly observe, enumerate, and compare every theoretically relevant quantity (induced partitions, exact finite-ecology decompositions, population-level selection trajectories) against the theorems.

The laboratory world states are languages or language groups drawn from aligned multilingual corpora of three widely translated texts: \emph{Alice's Adventures in Wonderland}, Dante's \emph{Commedia}, and the \emph{Communist Manifesto}. The off-ecology probe uses the Voynich manuscript through an EVA transliteration from Rene Zandbergen's digital archive. Specific editions and digital sources are listed in \cref{app:corpora}. In the neural experiments, the observables are behavioural distance matrices, thresholded induced partitions, held-out token losses, and population-level selection trajectories. At the exact level, we collapse the same held-out corpora into finite empirical ecologies whose world states are languages and whose contexts are short prefix-length conditions, so we can evaluate the theorem quantities directly rather than only through SGD-trained proxies.

This distinction between an exact empirical ecology and a learned neural approximation also determines how the empirical results should be read. Some figures report theorem quantities directly, evaluated either in finite synthetic ecologies or in held-out empirical ecologies. Others report the behaviour of trained models relative to those same quantities, and therefore include the additional effects of optimization error, finite capacity, and finite-sample noise. The model organism validates the theoretical machinery in a regime where all quantities are observable; the predictions for production models (\cref{sec:conclusion}) necessarily rely on proxies.

\section{The Autoregressive Task Ecology}\label{sec:ecology}

We use only a limited part of the framework of \citet{dallariva_2026}. In
that setting, one starts with a finite world-state space $W$, a task
ecology $\mu$ over functions on $W$, and an encoding $p$ that may merge
world states. The ecology induces an equivalence relation on $W$: two
states are equivalent when the tasks sampled from $\mu$ do not
distinguish them. An encoding is \emph{ecologically veridical} when it
merges only states that are equivalent in that sense. The static theory
then identifies those veridical encodings as the zero-excess solutions,
while the dynamic theory adds conditional evolutionary convergence when
a genuine reproduction--selection--mutation process is present.

For an autoregressive language model with frozen weights, those objects
take the following form.

\subsection{World States and Linguistic Accessibility}

\begin{definition}[World states]
Let $W$ be a finite set of \emph{world states}: latent configurations of reality that are relevant to the agent's task ecology, equipped with a prior distribution $\pi$ with $\pi(w) > 0$ for all $w \in W$. Each $w \in W$ determines a joint distribution over observable texts.
\end{definition}

We do not require $W$ to include all possible configurations of reality, only those that the agent's task ecology may query. The finiteness assumption matches the finite-state setup of \citet{dallariva_2026} and holds because any practical task ecology distinguishes only finitely many states.

For language models, $W$ includes cultural and informational states alongside physical ones. Moreover, LLMs are now among the agents that produce such states: model-generated text enters future training corpora, model-written code becomes infrastructure, model outputs reshape what is ``true'' about the informational environment. $W$ is therefore not exogenous to the population of models whose veridicality we study; it is partially co-constructed by them. At any snapshot in time, we may fix $W$ and apply the framework's static results (Thms.~\ref{thm:llm-veridicality} and~\ref{thm:min-complexity}). But the interpretation of ecological veridicality must acknowledge that the target reality is itself a moving object shaped by the models that are veridical to it. We return to this point in \cref{sec:conclusion} under the heading of niche construction.

Let $V$ be a finite token vocabulary, let $V^*$ denote the set of all finite token sequences over $V$, and let $\Delta(V)$ denote the simplex of probability distributions on $V$.

\begin{definition}[Text distribution conditioned on world state]
For each $w \in W$ and each finite token sequence (context) $c \in V^*$, let $P_w(\cdot \mid c) \in \Delta(V)$ denote the conditional distribution of the next token given context $c$ when the world state is $w$.
\end{definition}

\begin{definition}[Linguistic equivalence]\label{def:ling-equiv}
Two world states $w_1, w_2 \in W$ are \emph{linguistically equivalent}, written $w_1 \approx_L w_2$, if
\[
P_{w_1}(\cdot \mid c) = P_{w_2}(\cdot \mid c)
\quad \text{for all } c \in V^*.
\]
That is, no text context can distinguish them.
\end{definition}

\begin{remark}
Linguistic equivalence is an equivalence relation on $W$ (reflexive, symmetric, and transitive, the last by transitivity of equality). The equivalence classes $[w]_L$ partition $W$ into groups of states that are indistinguishable through text. These classes are at least as coarse as the task-equivalence classes $[w]_\mu$, i.e.\ the equivalence classes induced by the task ecology $\mu$ in the framework of \citet{dallariva_2026}, and in general strictly coarser: an embodied agent with non-linguistic sensory channels may separate states that are linguistically equivalent.
\end{remark}

\subsection{Contexts as Tasks}

In this subsection, we translate the ecological-veridicality framework into the autoregressive setting. We treat next-token prediction as a task ecology in the precise sense needed by \citet{dallariva_2026}, and the resulting objective admits the same kind of exact excess-loss decomposition. Once that translation is in place, ecological veridicality becomes a direct statement about the standard token-level training loss.

The ingredients of the decomposition are standard information-theoretic facts: Bayes-optimal prediction under log-loss is given by conditional mixtures, and the excess above the entropy floor expands into conditional KL or Jensen--Shannon terms. What is new here is their assembly for the autoregressive world-state setting.

\begin{definition}[Vector context-task]\label{def:context-task}
A \emph{context-task} is a vector-valued function $f_c \colon W \to \R^{|V|}$ defined by a context $c \in V^*$:
\[
f_c(w) = P_w(\cdot \mid c),
\]
the next-token distribution in world state $w$.
\end{definition}

\begin{definition}[Training task ecology]\label{def:train-ecology}
Let $D$ be a distribution over context-target pairs $(c, v)$, where $c \in V^*$ is a context and $v \in V$ is the next-token target, and let $D_C$ be its marginal over contexts. The \emph{training task ecology} is the pushforward measure
\[
\mu_D = D_C \circ f^{-1}_{(\cdot)},
\]
i.e.\ the distribution over vector tasks obtained by sampling a context $c$ from $D_C$ and mapping it to the corresponding task $f_c$.
\end{definition}

For the induced task ecology and the excess-loss decomposition below, the full pair distribution $D$ matters only through its context marginal $D_C$. The next-token law is supplied separately by the world-conditioned distributions $P_w(\cdot \mid c)$.

\begin{definition}[Token log-loss of an encoding]\label{def:ce-risk}
For an encoding $p \colon W \to X$ into an abstract code space $X$, and a decoder $q \colon X \times V^* \to \Delta(V)$, define the expected next-token cross-entropy under the training distribution $D$ by
\[
L_D(p,q)
:= \E_{w \sim \pi,\, c \sim D_C,\, v \sim P_w(\cdot \mid c)}
   \bigl[-\log q(v \mid p(w),c)\bigr].
\]
Equivalently,
\[
L_D(p,q)
= \E_{w \sim \pi,\, c \sim D_C}
  \bigl[\CE(P_w(\cdot \mid c),\, q(p(w),c))\bigr],
\]
where $\CE(P,Q) = H(P) + \KL(P \| Q)$ is cross-entropy, $H(P)$ is Shannon entropy, and $\KL(P\|Q)$ is Kullback--Leibler divergence.
\end{definition}

In the entropy and mutual-information identities below, we write $(W,C,Y)$ for the random world state, context, and next token generated by
\[
W \sim \pi,
\qquad
C \sim D_C,
\qquad
Y \sim P_W(\cdot \mid C),
\]
and set $X := p(W)$.

\begin{theorem}[Optimal decoder and exact excess-loss decomposition]\label{thm:ce-decomposition}
Fix an encoding $p \colon W \to X$, let $X = p(W)$, and write $C_x := \{w \in W : p(w)=x\}$ for the cell of code~$x$. For each non-empty cell define
\[
\pi_x := \sum_{w \in C_x} \pi(w),
\qquad
\alpha_x(w) := \pi(w)/\pi_x,
\]
and the cell-average next-token distribution
\[
\bar P_x(\cdot \mid c)
:= \sum_{w \in C_x} \alpha_x(w)\, P_w(\cdot \mid c).
\]

Then:
\begin{enumerate}[label=(\alph*)]
\item The Bayes-optimal decoder for $p$ is
$q_p^*(\cdot \mid x,c) = \bar P_x(\cdot \mid c)$.

\item The optimal loss attainable with encoding $p$, denoted
$L_D^*(p) := \inf_q L_D(p,q)$, satisfies
\[
L_D^*(p) = H(Y \mid C,X),
\]
and admits the exact decomposition
\[
\begin{aligned}
L_D^*(p)
&= H(Y \mid C,W) + I(Y;W \mid C,X) \\
&= H(Y \mid C,W)
  + \E_{c \sim D_C}
    \biggl[\sum_x \pi_x\,
    \JS_{\alpha_x}\bigl(\{P_w(\cdot \mid c)\}_{w \in C_x}\bigr)\biggr].
\end{aligned}
\]
where all entropies and mutual informations are taken under the joint law induced by $\pi$, $D_C$, and the conditional token distributions $P_w(\cdot \mid c)$, and where $\JS_{\alpha_x}$ is the weighted Jensen--Shannon divergence inside cell~$C_x$, i.e.\ the $\alpha_x$-weighted average of $\KL(P_w(\cdot\mid c)\|\bar P_x(\cdot\mid c))$ over $w \in C_x$.

\item Consequently, the excess loss above the irreducible entropy floor, $L_D^*(p) - H(Y \mid C,W)$, vanishes if and only if every cell of $p$ contains only training-equivalent states, i.e.\ whenever $p(w_1)=p(w_2)$, we have
\[
P_{w_1}(\cdot \mid c) = P_{w_2}(\cdot \mid c)
\]
for $D_C$-almost every $c$. If $D_C$ separates all points, equality requires $p$ to be injective on~$W$.
\end{enumerate}
\end{theorem}

\begin{proof}
For fixed $x$ and $c$, the contribution to $L_D(p,q)$ from code $x$ is
\[
\sum_{w \in C_x} \pi(w)\,
\CE\bigl(P_w(\cdot \mid c),\, q(x,c)\bigr).
\]
Using $\CE(P,Q)=H(P)+\KL(P\|Q)$, this equals
\[
\sum_{w \in C_x} \pi(w)\, H(P_w(\cdot \mid c))
+ \pi_x \sum_{w \in C_x} \alpha_x(w)\,
  \KL\bigl(P_w(\cdot \mid c)\|q(x,c)\bigr).
\]
The first term is independent of $q$, and the second is minimized at the mixture $q(x,c)=\bar P_x(\cdot \mid c)$, proving~(a). Substituting this optimizer yields
\[
L_D^*(p)
= \E_{w,c,v}\bigl[-\log P(Y=v \mid X=p(w), C=c)\bigr]
= H(Y \mid C,X),
\]
which is the first identity in~(b). The standard chain rule gives
\[
H(Y \mid C,X) = H(Y \mid C,X,W) + I(Y;W \mid C,X).
\]
Since $X=p(W)$ is a deterministic function of~$W$, conditioning on $X$ in addition to~$W$ adds no information, so $H(Y\mid C,X,W)=H(Y\mid C,W)$. Therefore
\[
H(Y \mid C,X) = H(Y \mid C,W) + I(Y;W \mid C,X).
\]
Expanding the conditional mutual information cell-by-cell gives the weighted Jensen--Shannon form in~(b):
\[
I(Y;W \mid C,X)
= \E_{c \sim D_C}
  \biggl[\sum_x \pi_x
  \sum_{w \in C_x}\alpha_x(w)\,
  \KL\bigl(P_w(\cdot \mid c)\|\bar P_x(\cdot \mid c)\bigr)\biggr].
\]
For~(c), each weighted Jensen--Shannon term is non-negative and is zero iff all distributions in that cell agree. Hence the total excess is zero iff for every $x$ and $D_C$-almost every $c$, the family $\{P_w(\cdot \mid c)\}_{w \in C_x}$ is constant. If $D_C$ separates all points, no two distinct states can satisfy this, so every zero-excess encoding must be injective.
\end{proof}

\begin{figure}[t]
\centering
\includegraphics[width=0.78\linewidth]{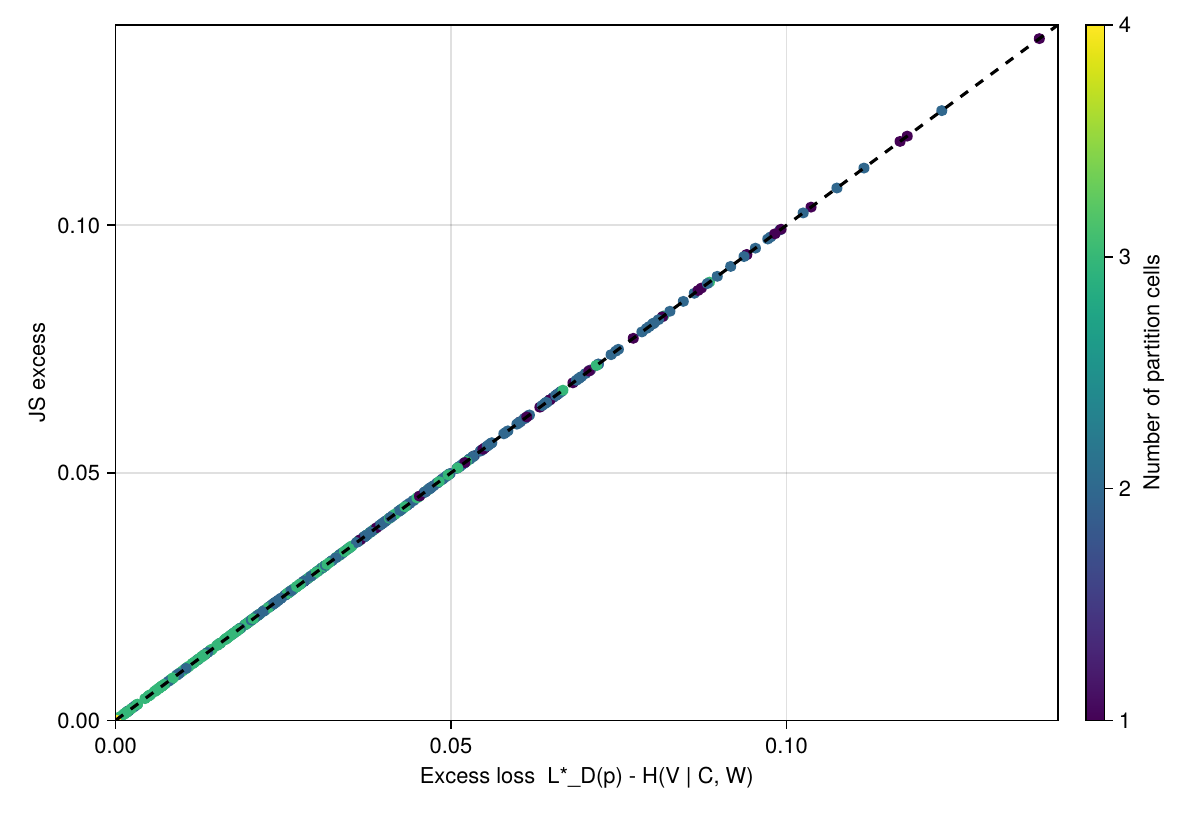}
\caption{Exact finite-ecology calibration of \thmref{thm:ce-decomposition}. Each point is a discrete ecology/partition pair from the synthetic sweep. The excess loss coincides exactly with the Jensen--Shannon excess term, giving the diagonal identity predicted by the theorem.}
\label{fig:exp0-decomposition}
\end{figure}

\begin{figure}[t]
\centering
\includegraphics[width=\linewidth]{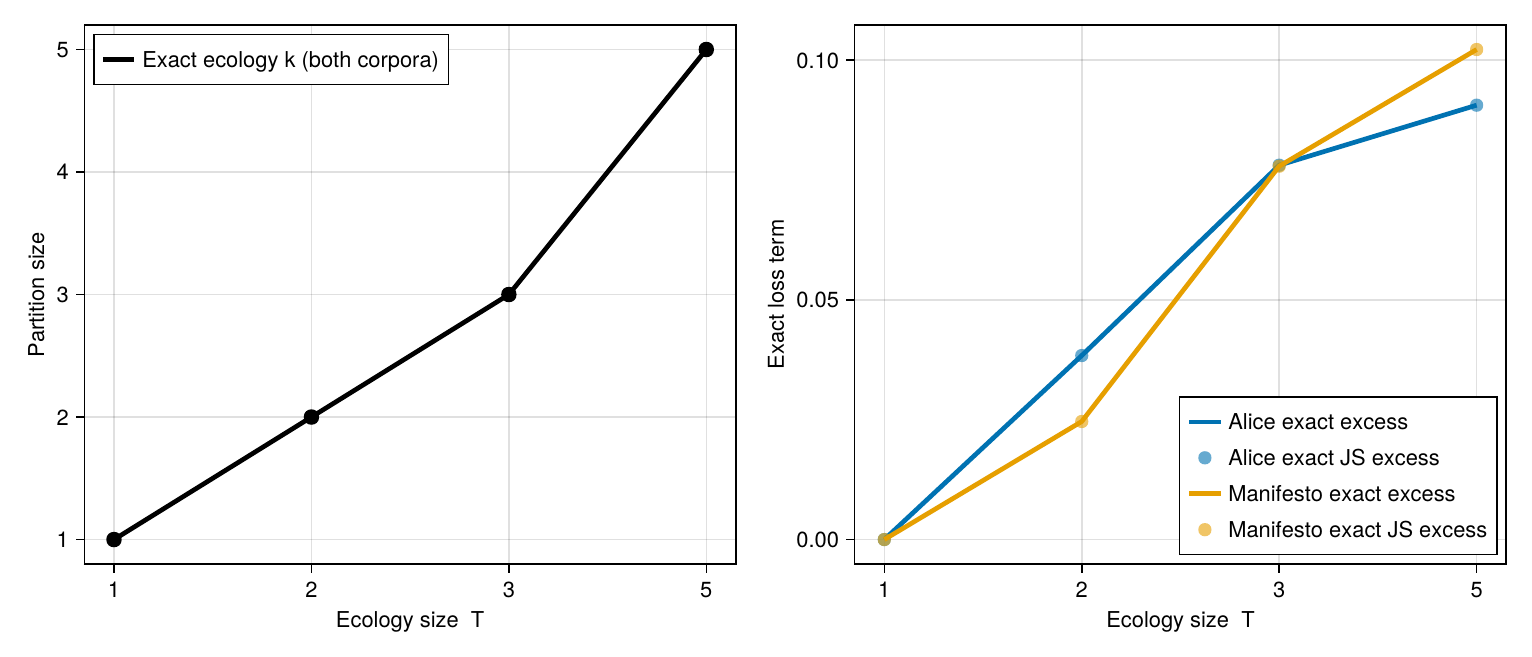}
\caption{Empirical-corpus corroboration of the static theory. Left: the exact empirical ecology built from held-out \emph{Alice} and \emph{Manifesto} corpora has partition size $k_{\mathrm{exact}}=1,2,3,5$ as the ecology expands. Right: on those same exact empirical ecologies, the measured excess loss coincides with the exact Jensen--Shannon excess term, as predicted by \thmref{thm:ce-decomposition}.}
\label{fig:exp1-static}
\end{figure}

\begin{definition}[Task distance under the training ecology]\label{def:text-distance}
The task distance under $\mu_D$ is defined via the squared Hellinger distance between the next-token distributions:
\[
\sigmasq_D(w_1, w_2)
= \E_{c \sim D_C}\bigl[\Hsq(P_{w_1}(\cdot\mid c),\, P_{w_2}(\cdot\mid c))\bigr],
\]
where $D_C$ is the marginal distribution of $D$ over contexts and $\Hsq(P, Q) = \tfrac{1}{2} \sum_v (\sqrt{P(v)} - \sqrt{Q(v)})^2$ is the squared Hellinger distance.
\end{definition}

\citet{dallariva_2026} defines task distance via expected squared difference of task values. The qualitative separation structure, i.e.\ which pairs of world states have $\sigmasq_D > 0$, depends only on whether $P_{w_1}(\cdot\mid c) \neq P_{w_2}(\cdot\mid c)$ on a set of positive $D_C$-measure, and is therefore identical under any divergence that vanishes exactly on equality. For the actual autoregressive objective, \thmref{thm:ce-decomposition} provides the exact loss statement directly: next-token cross-entropy is minimized exactly when the encoding preserves the same $D_C$-almost-everywhere equivalence classes. Hellinger is retained only as an auxiliary quantitative metric because it gives a Hilbert-space geometry (\cref{app:geometry}) and the same zero/nonzero separation structure. Thus we import the separation logic from \citet{dallariva_2026}, but recast the geometry in a Hellinger-based analogue; the main loss results below are proved directly for cross-entropy.

The squared Hellinger distance is an $\ell^2$ norm on square-root-transformed distributions: $\Hsq(P,Q) = \tfrac{1}{2}\|\sqrt{P} - \sqrt{Q}\|_2^2$. This preserves the Hilbert space structure needed for the geometric results in \cref{app:geometry} (the canonical embedding becomes $\Psi_D(w)(c) = \sqrt{P_w(\cdot\mid c)} \in L^2(D_C; \R^{|V|})$). For the present paper, the only unconditional comparison facts we use are the one-sided bound
\[
\Hsq(P,Q) \;\le\; -2\log\!\bigl(1 - \Hsq(P,Q)\bigr) \;\le\; \KL(P \| Q),
\]
where the second inequality is R\'{e}nyi monotonicity (here $D_\alpha$ denotes the R\'{e}nyi divergence of order $\alpha$, with $D_{1/2}(P\|Q) = -2\log \sum_v \sqrt{P(v)Q(v)}$ and $D_1(P\|Q)=\KL(P\|Q)$). Thus positive Hellinger separation implies positive KL separation, and small KL forces small Hellinger. The main text uses only these one-sided facts and the shared zero set $\Hsq(P,Q)=0 \Leftrightarrow P=Q \Leftrightarrow \KL(P\|Q)=0$; any stronger local equivalence is inessential here.

\begin{remark}[Scalar-coordinate version]
If one instead defines scalar tasks $f_{c,v}(w)=P_w(v\mid c)$, then $\E_{(c,v)\sim D}[(f_{c,v}(w_1)-f_{c,v}(w_2))^2] = \E_{c\sim D_C}\!\left[\sum_v D(v\mid c)\Delta_v(c)^2\right]$, with $\Delta_v(c)=P_{w_1}(v\mid c)-P_{w_2}(v\mid c)$. This equals the unweighted $\ell^2$ norm only under additional assumptions on $D(v\mid c)$. The vector-task formulation avoids this mismatch.
\end{remark}

\begin{corollary}[Separation under the training ecology]\label{cor:text-separation}
The training ecology $\mu_D$ separates $w_1$ from $w_2$ (in the sense of \citet[Definition~3.3]{dallariva_2026}) if and only if
\[
D_C\bigl(\{c \in V^* : P_{w_1}(\cdot\mid c) \neq P_{w_2}(\cdot\mid c)\}\bigr) > 0.
\]
That is, there exists a set of training contexts of positive measure under which the next-token distributions differ.
\end{corollary}

\begin{proof}
By the definition of task distance under the training ecology,
\[
\sigmasq_D(w_1,w_2)
= \E_{c\sim D_C}\!\left[\Hsq\bigl(P_{w_1}(\cdot\mid c),P_{w_2}(\cdot\mid c)\bigr)\right].
\]
The squared Hellinger distance is nonnegative and vanishes exactly when its two arguments are equal. Hence the expectation is strictly positive if and only if the integrand is strictly positive on a set of positive $D_C$-measure. This is equivalent to saying that
\[
P_{w_1}(\cdot\mid c)\neq P_{w_2}(\cdot\mid c)
\]
for a set of contexts $c$ of positive $D_C$-measure, which is exactly the stated condition.
\end{proof}

\begin{definition}[Textual separation margin]\label{def:text-margin}
When $\mu_D$ separates all points of $W$:
\[
\delta_D = \min_{w_1 \neq w_2} \sigmasq_D(w_1, w_2) > 0.
\]
\end{definition}

\begin{remark}[Quantitative connection to training loss]
\thmref{thm:ce-decomposition} already gives the exact zero/nonzero characterisation for the token-level cross-entropy objective. Hellinger serves only an auxiliary role: it provides a geometry on world states and a quantitative surrogate for how strongly a pair is separated. The unconditional facts we use are only that $\Hsq(P,Q) \le \KL(P \| Q)$ and that both vanish iff $P=Q$. Thus Hellinger separation implies KL separation, and small KL implies small $\Hsq$. If one also works locally away from the boundary of the simplex, then the two divergences are second-order equivalent near equality, with $\KL(P\|Q)=4\Hsq(P,Q)+o(\Hsq)$ under our convention for $\Hsq$.
\end{remark}

\subsection{Bounding the Linguistic Separation}

The previous subsection defined the training ecology induced by a corpus. We now relate that ecology to the larger space of distinctions that are in principle expressible in language at all. This matters because the training corpus can only separate pairs that are both linguistically distinguishable and actually probed by contexts of positive training measure.

\begin{proposition}[Linguistic equivalence bounds the text ecology]\label{prop:ling-bound}
For all $w_1, w_2 \in W$:
\begin{enumerate}[label=(\alph*)]
\item If $w_1 \approx_L w_2$ then $\sigmasq_D(w_1, w_2) = 0$ for every training distribution~$D$.

\item If $w_1 \not\approx_L w_2$, then there exists a training distribution $D$ such that $\sigmasq_D(w_1, w_2) > 0$.

\item Therefore, $[w]_{\mu_D} \supseteq [w]_L$ for every $D$. Equality holds whenever, for every pair $w_1 \not\approx_L w_2$, the context marginal $D_C$ assigns positive mass to at least one separating context $c$ with $P_{w_1}(\cdot\mid c) \neq P_{w_2}(\cdot\mid c)$. Since $W$ is finite, this requires positive mass only on a finite witness set of such contexts; full support on $V^*$ is a stronger idealization, not a necessity.
\end{enumerate}
\end{proposition}

\begin{proof}
(a) If $P_{w_1}(\cdot\mid c) = P_{w_2}(\cdot\mid c)$ for all $c$, every integrand vanishes.
(b) By $w_1 \not\approx_L w_2$, $\exists\, c^*$ with $P_{w_1}(\cdot\mid c^*) \neq P_{w_2}(\cdot\mid c^*)$. Let $D_C$ assign mass $\varepsilon > 0$ to $c^*$. Then $\sigmasq_D > 0$.
(c) Part~(a) implies $[w]_{\mu_D} \supseteq [w]_L$ for every $D$: linguistic equivalence forces zero ecology distance under every training distribution. For the converse under the stated witness condition, take any pair $w_1 \not\approx_L w_2$. By hypothesis there is a separating context $c^*$ with $D_C(c^*) > 0$. Then part~(b) gives $\sigmasq_D(w_1,w_2) > 0$, so $w_1$ and $w_2$ cannot lie in the same $\mu_D$-equivalence class. Hence no $\mu_D$-class can strictly contain multiple linguistic classes, and $[w]_{\mu_D} = [w]_L$. Because $W$ is finite, only finitely many non-linguistically-equivalent pairs exist, so only finitely many witness contexts are needed. Full support on $V^*$ implies this condition automatically, but is stronger than what the argument actually uses.
\end{proof}

\begin{proposition}[Ecology expansion refines equivalence]\label{prop:ecology-expansion}
Let $\mu' = (1-\alpha)\mu_D + \alpha \nu$ for some $\alpha \in (0,1]$ and any additional task distribution $\nu$. Then for all $w_1,w_2$:
\[
\sigmasq_{\mu'}(w_1,w_2) = (1-\alpha)\sigmasq_D(w_1,w_2) + \alpha\, \sigmasq_\nu(w_1,w_2).
\]
Hence $[w]_{\mu'} \subseteq [w]_{\mu_D}$: adding task families can split existing equivalence classes but cannot merge previously separated states.
\end{proposition}

\begin{proof}
For each pair $(w_1,w_2)$,
\[
\sigmasq_{\mu'}(w_1,w_2)
= \E_{t\sim \mu'}\,\E_{q\sim D_t}[d_t(P^t_{w_1}(\cdot\mid q),P^t_{w_2}(\cdot\mid q))^2].
\]
Since $\mu'=(1-\alpha)\mu_D+\alpha\nu$, linearity of expectation gives the displayed interpolation identity. If $\sigmasq_D(w_1,w_2)>0$, then the first term contributes $(1-\alpha)\sigmasq_D(w_1,w_2)>0$, so every pair separated by~$\mu_D$ remains separated by~$\mu'$. Hence $\mu'$ can split existing equivalence classes but cannot merge previously separated states.
\end{proof}

\begin{interpretation}
The training corpus determines which linguistically accessible distinctions are actually separated. A corpus that never includes contexts probing the difference between $w_1$ and $w_2$ leaves them merged, even if they are linguistically distinguishable. The textual separation margin $\delta_D$ is a property of the corpus, not of language in the abstract.
\end{interpretation}

\section{The Frozen Transformer as a Fixed Encoding}\label{sec:frozen}

The ecological-veridicality theorems require a single encoding held fixed across the tasks sampled from the ecology. In the present setting, the corresponding architectural question is whether a transformer deploys one task-invariant implementation whose behaviour varies only with the input context, or whether the implementation itself changes across tasks. Cognitive impenetrability plays that role below.

\subsection{Dense Transformers}

\begin{definition}[Frozen transformer implementation]\label{def:dense-encoding}
A \emph{dense transformer} with frozen weight vector $\theta \in \Theta$, where $\Theta$ denotes the parameter space of the architecture, defines a function $F_\theta \colon V^* \to \Delta(V)$ mapping every context $c$ to a distribution over the next token. Here, $\theta$ is the \emph{implementation-level parameterisation}. The representational object of interest is not $\theta$ itself, nor any particular hidden-state tensor, but the ecology-relative induced state encoding derived from the model's behaviour over world states, defined precisely below.
\end{definition}

\begin{proposition}[Cognitive impenetrability of frozen dense transformers]\label{prop:dense-impen}
The implementation $\theta$ induces a cognitively impenetrable state encoding:
\begin{enumerate}[label=(\alph*)]
\item $\theta$ is fixed across all tasks (contexts).
\item Different tasks produce different outputs only because they produce different contexts $c$, processed by the same fixed function $F_\theta$.
\item Any state distinctions available to the model are therefore induced by a single fixed map $F_\theta$, with task variation entering through $c$ alone.
\end{enumerate}
\end{proposition}

\begin{proof}
Part~(a) is immediate from the frozen-weights assumption: the same parameter vector~$\theta$ is used for every input context. Part~(b) then follows because the output distribution on any task is computed by the single map $F_\theta$, evaluated at different contexts~$c$. Part~(c) is just the corresponding representation-level statement: any distinguishability the model exhibits must be induced by that same fixed implementation, with task variation entering only through the context.
\end{proof}

A transformer computes intermediate hidden states $h_l(c)$ at each layer. These depend on $c$ and hence vary across inputs. They are not the ``encoding'' in the sense of \citet{dallariva_2026}, and current empirical work does not give a unique, theory-independent way of identifying \emph{the} representation of an LLM from weights or activations alone. Probes, representational-similarity methods, and interventions provide partial empirical access, but they do not eliminate the need for abstraction. For the formal theory, the relevant object is therefore the operational equivalence relation over world states induced by the model's behaviour under a probe repertoire. Hidden states are possible empirical windows onto that object, not the object itself.

The Transformer Circuits framework gives this a useful architectural reading: in a frozen transformer, attention heads and MLP blocks are additive readers and writers on a shared residual stream \citep{elhage_etal_2021_framework}. Different contexts can recruit different circuit compositions, but they still do so through one fixed implementation acting on one shared state space. That is the mechanism-level analogue of the impenetrability condition used here.

\subsection{Mixture-of-Experts Transformers}

\begin{definition}[MoE transformer]\label{def:moe}
A \emph{Mixture-of-Experts transformer} with frozen weights $\theta = (\theta_{\mathrm{shared}}, \theta_1, \ldots, \theta_E, \theta_{\mathrm{router}})$ defines a routing function $r \colon V^* \to 2^{[E]}$, where $[E] := \{1,\ldots,E\}$ and $2^{[E]}$ is its power set, with $|r(c)| = k$ for fixed $k < E$, determined by $\theta_{\mathrm{router}}$. Thus $r(c)$ is the subset of experts activated by context $c$. For each input $c$, the active parameter set is $\theta(c) = (\theta_{\mathrm{shared}}, \{\theta_e : e \in r(c)\})$.
\end{definition}

\begin{proposition}[Cognitive impenetrability of frozen MoE transformers]\label{prop:moe-impen}
A frozen MoE transformer is cognitively impenetrable: the full weight vector $\theta$ is fixed at inference, the routing function $r$ is determined by $\theta_{\mathrm{router}}$ and the input $c$ (not by a task identifier), and the mapping $F_\theta \colon V^* \to \Delta(V)$ is a single fixed function.
\end{proposition}

\begin{proof}
The full parameter tuple $(\theta_{\mathrm{shared}},\theta_1,\ldots,\theta_E,\theta_{\mathrm{router}})$ is fixed at inference. For each input~$c$, the router computes $r(c)$ from~$c$ using the frozen parameters $\theta_{\mathrm{router}}$; there is no independent task-specific parameter update. Consequently the overall input-output map $F_\theta$ is a single fixed function of~$c$, even though different contexts activate different expert subsets.
\end{proof}

Note that MoE routing is input-dependent ($r(c)$ varies with $c$), but this is true of any non-trivial function. The relevant distinction is that $r$ does not receive a task identifier as input. Per-task fine-tuning, by contrast, changes $\theta$ itself depending on the task, which constitutes cognitive penetration (\cref{sec:penetrability}).

\subsection{In-Context Learning}

To compare transformers with the world-state encodings of \citet{dallariva_2026}, we must connect latent world states to textual evidence presented to the model. The next two definitions make that interface explicit and then define the induced equivalence relation on world states generated by the model's behaviour on those prompts.

\begin{definition}[World-text interface]
Fix an interface map $\mathrm{obs} \colon W \to V^*$ that provides textual evidence for world state $w$. For probe context $c$, the model is queried on $c \oplus \mathrm{obs}(w)$, where $\oplus$ denotes sequence concatenation.
\end{definition}

\begin{definition}[Readout repertoire]\label{def:repertoire}
For a frozen transformer implementation $\theta$, define the \emph{separation set}:
\[
S(\theta) := \{(w_1, w_2) : \exists\, c \in V^*\ \text{s.t.}\ F_\theta(c \oplus \mathrm{obs}(w_1)) \neq F_\theta(c \oplus \mathrm{obs}(w_2))\},
\]
the set of world-state pairs that $\theta$ can distinguish under some context.
\end{definition}

\begin{proposition}[ICL does not expand separation]\label{prop:icl-bound}
$S(\theta)$ is determined by $\theta$ alone. In-context learning selects which context $c$ to use, thereby selecting which element of the readout repertoire to deploy, but does not enlarge $S(\theta)$.
\end{proposition}

\begin{proof}
$F_\theta$ is a fixed function determined by $\theta$. For any $c$, the outputs $F_\theta(c \oplus \mathrm{obs}(w_1))$ and $F_\theta(c \oplus \mathrm{obs}(w_2))$ are values of this fixed function. $S(\theta)$ is the union over all $c$ of the set of pairs distinguished by $F_\theta(c \oplus \cdot)$, which is determined by~$\theta$.
\end{proof}

In the circuit language of \citet{elhage_etal_2021_framework}, induction-style in-context learning is a concrete example of this bounded flexibility: prompts can alter which composed circuit is activated, but they do so through the same frozen QK/OV machinery. The prompt changes the readout path, not the underlying separation set available to the implementation.

\begin{corollary}[ICL and training-time veridicality]
If $(w_1,w_2) \notin S(\theta)$, then no prompting strategy can make the frozen model distinguish that pair. Whenever a deployment ecology assigns positive separation weight to such a pair, a strictly positive excess token loss is unavoidable for that frozen~$\theta$ (by \thmref{thm:ce-decomposition}(c)).
\end{corollary}

\begin{proof}
By \propref{prop:icl-bound}, in-context learning can only choose among contexts already available to the fixed implementation~$\theta$; it cannot enlarge~$S(\theta)$. Thus if $(w_1,w_2)\notin S(\theta)$, then
\[
F_\theta(c\oplus \mathrm{obs}(w_1)) = F_\theta(c\oplus \mathrm{obs}(w_2))
\]
for every prompt context~$c$, so no prompting strategy can separate the pair. If a deployment ecology nevertheless assigns positive separation weight to that pair, then the induced encoding merges a deployment-separated distinction, and \thmref{thm:ce-decomposition}(c) implies strictly positive excess token loss.
\end{proof}

\begin{definition}[Operational state encoding]\label{def:model-encoding}
Relative to the world-text interface $\mathrm{obs}$ and the context marginal $D_C$ under discussion, define an equivalence relation $\sim_{\theta,D}$ on $W$ by
\[
w_1 \sim_{\theta,D} w_2
\quad\text{iff}\quad
F_\theta(c \oplus \mathrm{obs}(w_1)) = F_\theta(c \oplus \mathrm{obs}(w_2))
\quad \text{for } D_C\text{-almost every } c.
\]
Let $p_{\theta,D} \colon W \to W/{\sim_{\theta,D}}$ map each world state to its $\sim_{\theta,D}$-equivalence class under the context marginal $D_C$. This induced partition is the abstract encoding that lets us transport the separation logic of \citet{dallariva_2026} into the present cross-entropy framework.
\end{definition}

The object $p_{\theta,D}$ is defined by the model's behavior on $D_C$-almost every context, so finite probing generally cannot reveal it directly in realistic production LLMs. Only laboratory settings with exhaustively enumerable context sets, such as the microgpt experiments in our model-organism study, allow exact recovery. For production models, we can only estimate coarse proxies for the induced partition from finite prompt families and observed next-token distributions.

\begin{remark}[Separation set vs.\ ecology-relative encoding]
The full readout repertoire $S(\theta)$ from \defref{def:repertoire} records which pairs are distinguishable by \emph{some} context in~$V^*$.
The ecology-relative partition~$\sim_{\theta,D}$ is coarser: if a pair is separated only on a $D_C$-null set of contexts, then $(w_1,w_2) \in S(\theta)$ but $w_1 \sim_{\theta,D} w_2$.
Thus $S(\theta)$ can be strictly larger than what the training or deployment ecology actually exposes.
This is the model-side analogue of \corref{cor:text-separation}: zero-measure distinguishing contexts do not affect the ecology-induced partition.
\end{remark}

Chain-of-thought prompting and scratchpads can improve performance within a frozen model by generating intermediate tokens that create longer, more informative contexts \citep{wei_etal_2022_cot,nye_etal_2021_scratchpads}, but they do not enlarge the underlying separation set $S(\theta)$. The deployment decoding gap (\defref{def:decoder-gap}) formalizes this distinction: such procedures reduce the gap between the Bayes-optimal decoder and the restricted deployment class, without changing the representational term.

\begin{definition}[Ecological excess token loss of a model]\label{def:model-risk}
Define
\[
\Delta_D(\theta) := L_D^*(\theta) - H(Y \mid C,W).
\]
Then $\Delta_D(\theta)=0$ if and only if $p_{\theta,D}$ is ecologically veridical, by \thmref{thm:ce-decomposition}(c).
\end{definition}

\subsection{Partial Penetrability: Per-Task Adaptation}\label{sec:penetrability}

\begin{definition}[Per-task adaptation]
A model with per-task adaptation uses weights $\theta + \Delta\theta_\tau$ when performing task $\tau$, where $\tau$ is a task index and $\Delta\theta_\tau$ is the task-specific parameter update (LoRA adapter, prefix tuning, or full fine-tuning).
\end{definition}

\begin{proposition}[Per-task adaptation is cognitive penetration]
A model with per-task adaptation does not satisfy the cognitive-impenetrability assumption of \citet{dallariva_2026}. The implementation changes with $\tau$, so the induced state encoding need not be fixed across tasks. Hoffman's FBT applies independently to each task.
\end{proposition}

\begin{proof}
The implementation used on task~$\tau$ is $\theta+\Delta\theta_\tau$, so different tasks need not be processed by the same input-output map. Hence the induced encoding is not fixed across tasks, violating the cognitive-impenetrability premise required by the static ecological-veridicality framework. Once the implementation itself varies with~$\tau$, Hoffman's fixed-benefit theorem applies only task by task, not to a single shared encoding.
\end{proof}

\begin{remark}[The penetrability spectrum]
This yields a formal spectrum:
\begin{enumerate}[label=(\alph*)]
\item Fully impenetrable (frozen $\theta$): the fixed-encoding premise needed for the static theorem of \citet[Theorem~4.1]{dallariva_2026} is satisfied.
\item Partially penetrable (shared $\theta$ + small $\Delta\theta_\tau$): the shared base still faces multi-task pressure, but the effective ecology seen by the model may differ from the frozen-weight idealisation. Analysing that regime requires additional assumptions not developed here.
\item Fully penetrable (independent $\theta_\tau$ per task): Hoffman's FBT regime.
\end{enumerate}
\end{remark}

\subsection{Framework Mapping}

We summarize the correspondence between the ecological-veridicality framework and the frozen-transformer setting below.

\begin{center}
\begin{tabular}{@{}ll@{}}
\toprule
Ecological-veridicality framework & Frozen Transformer \\
\midrule
World states $W$ & Latent world configurations \\
Encoding $p \colon W \to X$ & Induced state encoding $p_{\theta,D}$ from $\sim_{\theta,D}$ \\
Task $f \colon W \to \R^d$ & Context-task $f_c(w) = P_w(\cdot\mid c)$ \\
Task distribution $\mu$ & $\mu_D$ induced by $D_C$ over contexts \\
Readout $a_f \colon X \to \text{Actions}$ & Task-specific Bayes-optimal readout on $p_{\theta,D}$-cells \\
Cognitive impenetrability & Frozen weights $\theta$ at inference \\
Task distance $\sigmasq(w_1,w_2)$ & $\E_{c\sim D_C}[\Hsq(P_{w_1}(\cdot\mid c), P_{w_2}(\cdot\mid c))]$ \\
Separation margin $\dmu$ & $\delta_D = \min_{w_1\neq w_2} \sigmasq_D(w_1,w_2)$ \\
\bottomrule
\end{tabular}
\end{center}

\section{Static Optimality for LLM Encodings}\label{sec:static}

The previous section supplied the model-side object that plays the role of an encoding, namely the induced partition $p_{\theta,D}$. We can now ask the static question central to the paper: when does the actual next-token objective favor induced encodings that preserve exactly the distinctions required by the training ecology?

\begin{theorem}[Cross-entropy optimum and ecological veridicality]\label{thm:llm-veridicality}
For $\theta \in \Theta$, write
\[
L_D^*(\theta) := L_D^*(p_{\theta,D})
\]
for the Bayes-optimal next-token cross-entropy induced by $p_{\theta,D}$. Assume this objective attains its minimum on~$\Theta$, and let $\theta^* \in \argmin_{\theta \in \Theta} L_D^*(\theta)$. Then:
\begin{enumerate}[label=(\alph*)]
\item The irreducible minimum $H(Y \mid C,W)$ is attained by $\theta^*$ iff $p_{\theta^*,D}$ merges only $\mu_D$-equivalent states.

\item If $\mu_D$ separates all points of $W$ and $\Theta$ realises at least one injective encoding on $W$, then any minimizer $\theta^*$ is fully veridical (up to label symmetry).

\item If every $\theta \in \Theta$ merges at least one $\mu_D$-separated pair, then
\[
\inf_{\theta \in \Theta} L_D^*(\theta) > H(Y \mid C,W),
\]
so the model class is necessarily lossy relative to the training ecology.
\end{enumerate}
\end{theorem}

\begin{proof}
Apply \thmref{thm:ce-decomposition} to the induced encoding $p_{\theta,D}$. Part~(a) is exactly the zero-excess characterization. For~(b), if $\mu_D$ separates all points and some $\theta$ induces an injective encoding, then \thmref{thm:ce-decomposition}(c) shows that this encoding attains the entropy floor $H(Y \mid C,W)$. Hence every minimizer $\theta^*$ must also attain that floor, and under full separation \thmref{thm:ce-decomposition}(c) again implies that only injective encodings can do so, i.e.\ every minimizer is fully veridical up to relabelling of codes. For~(c), if every $\theta$ merges a $\mu_D$-separated pair, then no induced encoding can satisfy the $D_C$-almost-everywhere equality condition inside every cell, so the Jensen--Shannon excess term in \thmref{thm:ce-decomposition}(b) is strictly positive for every~$\theta$. Since $W$ is finite, $L_D^*(\theta)$ depends on $\theta$ only through the induced partition $p_{\theta,D}$, and there are at most $B(|W|)$ such partitions. The infimum is therefore a minimum over finitely many strictly positive values, so it lies strictly above the entropy floor $H(Y \mid C,W)$.
\end{proof}

\begin{remark}[Existence of minimisers]
The non-empty argmin assumption is standard. It holds, for example, for finite hypothesis classes, or more generally when $\Theta$ is compact and $\theta \mapsto L_D^*(\theta)$ is lower semicontinuous.
\end{remark}

\begin{remark}[Bell-number bound]
Here $B(|W|)$ denotes the Bell number, i.e.\ the number of set partitions of~$W$.
\end{remark}

\subsection{Finite-Class Generalization Guarantee}

The static theorem above characterizes the Bayes-optimal token-loss target under the training ecology, but it does not yet say when finite data and approximate empirical optimisation recover a veridical induced encoding. The next result provides a deliberately conservative learning-theoretic bridge: under a finite induced encoding class and bounded token losses, near-optimal empirical token loss for an oracle decoder objective is enough to force ecological veridicality whenever the veridicality gap is strictly positive.

This is a standard finite-class uniform-convergence argument specialized to the induced-encoding family: the proof is just Hoeffding concentration plus a union bound, applied to the token-loss gap defined by the ecological-veridicality criterion.

\begin{definition}[Empirical token log-loss]\label{def:emp-risk}
Draw iid triples $(w_1,c_1,v_1),\ldots,(w_n,c_n,v_n)$ from the joint distribution
\[
w \sim \pi,
\qquad
c \sim D_C,
\qquad
v \sim P_w(\cdot \mid c).
\]
For $\theta \in \Theta$, let $q_{\theta}^* := q_{p_{\theta,D}}^*$ be the Bayes-optimal decoder from \thmref{thm:ce-decomposition}. Define
\[
\bar L_n(\theta)
:= \frac{1}{n}\sum_{t=1}^n
   \bigl[-\log q_{\theta}^*(v_t \mid p_{\theta,D}(w_t),c_t)\bigr].
\]
\end{definition}

\begin{definition}[Technical assumption: finite induced encoding class]\label{ass:finite-class}
Let
\[
\mathcal{P}_\Theta := \{p_{\theta,D} : \theta \in \Theta\},
\]
and assume $M_\Theta := |\mathcal{P}_\Theta| < \infty$.
\end{definition}

\begin{definition}[Technical assumption: bounded per-task risk]\label{ass:bounded-risk}
Assume there exists $\tau \in (0,1)$ such that for every $\theta \in \Theta$ and every triple $(w,c,v)$ with positive sampling probability:
\[
q_{\theta}^*(v \mid p_{\theta,D}(w),c) \ge \tau.
\]
Then each token loss is bounded:
\[
0 \le -\log q_{\theta}^*(v \mid p_{\theta,D}(w),c) \le C_\tau,
\qquad
C_\tau := \log(1/\tau).
\]
\end{definition}

The next theorem is a finite-class concentration result over \emph{induced encodings paired with their Bayes-optimal decoders}. It is therefore not a theorem about SGD in transformer parameter space or about the trajectory of a single training run. More narrowly, it states when near-optimal empirical token loss for the oracle objective $\bar L_n$ certifies that the induced encoding is ecologically veridical.

\begin{theorem}[Finite-class certification from near-optimal token loss]\label{thm:sgd-conditional}
Assume:
\begin{enumerate}[label=(\roman*)]
\item There exists $\theta^v \in \Theta$ with $L_D^*(\theta^v)=H(Y \mid C,W)$ (equivalently: $p_{\theta^v,D}$ is ecologically veridical).

\item The learner outputs $\hat\theta$ with empirical optimisation error
\[
\bar L_n(\hat\theta) \le \inf_{\theta\in\Theta}\bar L_n(\theta) + \varepsilon_{\mathrm{opt}}.
\]

\item Technical assumptions~\ref{ass:finite-class} and~\ref{ass:bounded-risk} hold.
\end{enumerate}
Let $\rho$ be any probability distribution on $\mathcal{P}_\Theta$, fixed independently of the training sample, and write $p^v := p_{\theta^v,D}$. For each induced encoding $p \in \mathcal{P}_\Theta$, define the concentration radius
\[
\eta_\rho(p)
:= C_\tau\sqrt{\frac{\log(1/\rho(p))+\log(2/\alpha)}{2n}}.
\]
Define the smallest positive excess over non-veridical induced encodings by
\[
\gamma_D^{\CE}
:= \min_{p \in \mathcal{P}_\Theta:\, p\ \text{non-veridical}}
    \bigl(L_D^*(p)-H(Y \mid C,W)\bigr).
\]
For each non-veridical $p \in \mathcal{P}_\Theta$, write
\[
\Delta_D(p) := L_D^*(p)-H(Y \mid C,W),
\]
so $\Delta_D(p)\ge \gamma_D^{\CE}$.
For the veridical encoding, write
\[
N_v := \log(1/\rho(p^v))+\log(2/\alpha).
\]
For each non-veridical $p$, write
\[
N_p := \bigl(\log(1/\rho(p))+\log(2/\alpha)\bigr)
       \bigl(\gamma_D^{\CE}/\Delta_D(p)\bigr)^2.
\]
If $\varepsilon_{\mathrm{opt}} < \gamma_D^{\CE}$, then with probability at least $1-\alpha$:
\[
p_{\hat\theta,D} \text{ is ecologically veridical},
\]
provided
\[
n \ge \frac{2C_\tau^2}{(\gamma_D^{\CE}-\varepsilon_{\mathrm{opt}})^2}
      \max\!\Bigl\{N_v,\;
      \max_{p\ \text{non-veridical}} N_p\Bigr\}.
\]
\end{theorem}

\begin{proof}
For fixed $p \in \mathcal{P}_\Theta$, Hoeffding with range $[0,C_\tau]$ gives
\[
P\bigl(|\bar L_n(p)-L_D^*(p)|\ge\eta\bigr) \le 2\exp(-2n\eta^2/C_\tau^2).
\]
Setting $\eta=\eta_\rho(p)$ yields
\[
P\bigl(|\bar L_n(p)-L_D^*(p)|\ge\eta_\rho(p)\bigr)
\le \alpha\,\rho(p).
\]
Summing over $p \in \mathcal{P}_\Theta$ gives
\[
P\bigl(\exists\, p \in \mathcal{P}_\Theta:
|\bar L_n(p)-L_D^*(p)|\ge\eta_\rho(p)\bigr)
\le \alpha.
\]
Let $E_\rho$ denote the complementary event. On $E_\rho$, for the veridical partition $p^v$:
\[
\bar L_n(p^v) < H(Y \mid C,W) + \eta_\rho(p^v).
\]
For any non-veridical $p$:
\[
\bar L_n(p) > H(Y \mid C,W) + \Delta_D(p) - \eta_\rho(p).
\]
Therefore no non-veridical partition can satisfy the empirical near-optimality condition in~(ii) provided
\[
\Delta_D(p) - \eta_\rho(p) > \eta_\rho(p^v) + \varepsilon_{\mathrm{opt}}
\qquad
\text{for all non-veridical } p.
\]
It is enough to require
\[
\eta_\rho(p^v) \le (\gamma_D^{\CE}-\varepsilon_{\mathrm{opt}})/2
\]
and, for each non-veridical $p$,
\[
\eta_\rho(p)
\le \frac{\gamma_D^{\CE}-\varepsilon_{\mathrm{opt}}}{2}
\frac{\Delta_D(p)}{\gamma_D^{\CE}}.
\]
The first inequality is exactly the first term in the displayed sample-size bound. The second is exactly the second term. Under those two inequalities,
\[
\eta_\rho(p^v) + \varepsilon_{\mathrm{opt}}
\le (\gamma_D^{\CE}+\varepsilon_{\mathrm{opt}})/2
\]
and
\[
\eta_\rho(p)
\le \Delta_D(p) - (\gamma_D^{\CE}+\varepsilon_{\mathrm{opt}})/2,
\]
because $\Delta_D(p)\ge \gamma_D^{\CE}$. Hence
\[
\Delta_D(p) - \eta_\rho(p)
\ge (\gamma_D^{\CE}+\varepsilon_{\mathrm{opt}})/2
\ge \eta_\rho(p^v)+\varepsilon_{\mathrm{opt}},
\]
with strict inequality coming from the strict concentration inequalities on $E_\rho$. Thus $\hat\theta$ must induce a veridical partition on $E_\rho$, which has probability at least $1-\alpha$.
\end{proof}

\begin{corollary}[Uniform prior recovers the finite-class bound]\label{cor:sgd-uniform}
If $\rho(p)=1/M_\Theta$ for every $p \in \mathcal{P}_\Theta$, all concentration radii in \thmref{thm:sgd-conditional} become equal and the per-partition conditions collapse to a single bound. Since $\Delta_D(p)\ge \gamma_D^{\CE}$ for every non-veridical~$p$, the sample-size requirement reduces to
\[
n \ge \frac{2C_\tau^2}{(\gamma_D^{\CE}-\varepsilon_{\mathrm{opt}})^2}
    \bigl(\log(2M_\Theta)+\log(1/\alpha)\bigr).
\]
\end{corollary}

\begin{proof}
Under the uniform prior, $\log(1/\rho(p))=\log M_\Theta$ for every induced partition~$p$. The sample-size condition in \thmref{thm:sgd-conditional} therefore becomes
\[
n \ge \frac{2C_\tau^2}{(\Delta_D(p)-\varepsilon_{\mathrm{opt}})^2}
\bigl(\log 2+\log M_\Theta+\log(1/\alpha)\bigr)
\]
for every non-veridical~$p$. Since $\Delta_D(p)\ge \gamma_D^{\CE}$ on that set by definition of the ecological veridicality gap, it is enough to impose the displayed lower bound with $\Delta_D(p)$ replaced by~$\gamma_D^{\CE}$.
\end{proof}

\begin{corollary}[Conditional near-optimality in token loss]\label{cor:sgd-near}
Under assumptions (ii)--(iii) of \thmref{thm:sgd-conditional}, for any $\eta>0$, with probability at least $1-2M_\Theta\exp(-2n\eta^2/C_\tau^2)$:
\[
L_D^*(\hat\theta) \le \inf_{\theta\in\Theta} L_D^*(\theta) + \varepsilon_{\mathrm{opt}} + 2\eta.
\]
\end{corollary}

\begin{proof}
On $E_\eta$,
$L_D^*(\hat\theta)\le \bar L_n(\hat\theta)+\eta
\le \inf_{\theta}\bar L_n(\theta)+\varepsilon_{\mathrm{opt}}+\eta
\le \inf_{\theta}\bigl(L_D^*(\theta)+\eta\bigr)+\varepsilon_{\mathrm{opt}}+\eta
= \inf_{\theta}L_D^*(\theta)+\varepsilon_{\mathrm{opt}}+2\eta$.
The probability bound is exactly the concentration bound defining~$E_\eta$ in \thmref{thm:sgd-conditional}.
\end{proof}

An informative choice of $\rho$ is the entropic prior
\[
\rho_{\beta_0}(p)
\propto \exp\bigl(-\beta_0 H(p(W))\bigr),
\qquad
\beta_0 > 0.
\]
Then
\[
\log(1/\rho_{\beta_0}(p))
= \beta_0 H(p(W)) + \log Z_{\beta_0},
\]
where $Z_{\beta_0}$ is the normalizing constant. Under that choice, low-complexity induced partitions receive larger mass and therefore tighter concentration radii. If the model class contains a minimum-complexity veridical partition, its contribution is governed by $\beta_0 I^*(\mu_D)+\log Z_{\beta_0}$ from \thmref{thm:min-complexity}. The exact sample bound depends on the full maximum over non-veridical partitions and cannot in general be reduced to the gap-achieving partition alone without extra structure relating $\Delta_D(p)$ to $H(p(W))$.

The theorem is a uniform-convergence result over induced encodings, not a statement about SGD on transformer parameter space. Unlike the earlier ecological-risk formulation, the objective is the actual token-level log-loss, but each induced encoding $p_{\theta,D}$ is paired with its Bayes-optimal decoder $q_{\theta}^*$. The decomposition therefore separates representation choice from decoder optimality, and within those, optimisation error $\varepsilon_{\mathrm{opt}}$ from statistical error $\eta$. The gap between the Bayes-optimal decoder and the decoder a trained transformer actually implements is absorbed into the optimisation idealisation; \defref{def:decoder-gap} below isolates that term explicitly. The finite induced class assumption holds automatically since $W$ is finite, but $M_\Theta$ can reach the Bell number $B(|W|)$, which grows super-exponentially ($B(20) \approx 5.2 \times 10^{13}$). The bound is therefore mainly conceptual unless the effective induced class is far smaller than the worst-case partition count and $\gamma_D^{\CE}$ is not too small. The entropy $H(p(W))$ plays three roles in the framework: it is the minimum-complexity target (\thmref{thm:min-complexity}), the explicit simplicity term in $J_{D,\beta}$ below, and, under an entropic prior, the statistical price of certifying a partition from finite data.

\begin{definition}[Deployment decoder class and decoding gap]\label{def:decoder-gap}
Fix a nonempty class $\mathcal{Q}_{\mathrm{dep}}$ of admissible deployment-time decoders
\[
q \colon X \times V^* \to \Delta(V).
\]
For $\theta \in \Theta$, define the best deployment-realizable token loss by
\[
L_D^{\mathcal{Q}_{\mathrm{dep}}}(\theta)
:= \inf_{q \in \mathcal{Q}_{\mathrm{dep}}}
   L_D(p_{\theta,D},q),
\]
and the corresponding deployment decoding gap by
\[
\Gamma_D^{\mathcal{Q}_{\mathrm{dep}}}(\theta)
:= L_D^{\mathcal{Q}_{\mathrm{dep}}}(\theta) - L_D^*(\theta).
\]
\end{definition}

\begin{proposition}[Representational excess plus deployment decoding gap]\label{prop:decoder-gap}
For every $\theta \in \Theta$ and every nonempty deployment decoder class $\mathcal{Q}_{\mathrm{dep}}$:
\begin{enumerate}[label=(\alph*)]
\item The deployment decoding gap is nonnegative:
\[
\Gamma_D^{\mathcal{Q}_{\mathrm{dep}}}(\theta)\ge 0.
\]

\item The best deployment-realizable loss decomposes as
\[
L_D^{\mathcal{Q}_{\mathrm{dep}}}(\theta)
= H(Y \mid C,W)
  + \Delta_D(\theta)
  + \Gamma_D^{\mathcal{Q}_{\mathrm{dep}}}(\theta).
\]

\item If $\mathcal{Q}_{\mathrm{dep}}$ contains a Bayes-optimal decoder for $p_{\theta,D}$, then
\[
\Gamma_D^{\mathcal{Q}_{\mathrm{dep}}}(\theta)=0.
\]
\end{enumerate}
\end{proposition}

\begin{proof}
Because $\mathcal{Q}_{\mathrm{dep}}$ is a subset of the class of all decoders, we have
\[
L_D^{\mathcal{Q}_{\mathrm{dep}}}(\theta)
\ge L_D^*(\theta),
\]
which gives~(a). Part~(b) follows by adding and subtracting $L_D^*(\theta)$ and then using the definition
\[
\Delta_D(\theta)=L_D^*(\theta)-H(Y \mid C,W).
\]
For~(c), if $q_{\theta}^* \in \mathcal{Q}_{\mathrm{dep}}$ attains $L_D^*(\theta)$, then
\[
L_D^{\mathcal{Q}_{\mathrm{dep}}}(\theta)
\le L_D(p_{\theta,D},q_{\theta}^*)
= L_D^*(\theta).
\]
Combined with~(a), this yields equality and hence $\Gamma_D^{\mathcal{Q}_{\mathrm{dep}}}(\theta)=0$.
\end{proof}

This isolates the missing computational term cleanly. The finite-class theorem above controls the representational term $\Delta_D(\theta)$ and the statistical error of the oracle objective, but it says nothing about $\Gamma_D^{\mathcal{Q}_{\mathrm{dep}}}(\theta)$ for realistic deployment inference classes. Bounding that term for concrete transformer inference regimes is the joint ecology-computation problem left open here. \Cref{app:decoder-gap} records only the basic monotonicity facts needed for that separation.

\subsection{Capacity Criterion}\label{sec:capacity}

Define ecological complexity $k_D := |W/{\sim_{\mu_D}}|$.

\begin{proposition}[Capacity criterion for non-lossy versus lossy]\label{prop:capacity-criterion}
For a model class $\Theta$ with induced encodings $\{p_{\theta,D} : \theta \in \Theta\}$:
\begin{enumerate}[label=(\alph*)]
\item If there exists $\theta$ such that $p_{\theta,D}$ assigns distinct codes to distinct $\mu_D$-equivalence classes (equivalently: ${\sim_{\theta,D}}$ refines ${\sim_{\mu_D}}$), then the non-lossy regime is feasible and the entropy floor $H(Y \mid C,W)$ is attainable.

\item If no $\theta \in \Theta$ separates the $\mu_D$-equivalence classes in that sense, then the problem is necessarily lossy and $\inf_{\theta \in \Theta} L_D^*(\theta) > H(Y \mid C,W)$.
\end{enumerate}
\end{proposition}

\begin{proof}
Part~(a) follows from \thmref{thm:ce-decomposition}(c): separating the $\mu_D$-equivalence classes is exactly what is required for zero excess. For~(b), if no $\theta \in \Theta$ separates the $\mu_D$-equivalence classes, then every induced encoding $p_{\theta,D}$ merges at least one $\mu_D$-separated pair. \thmref{thm:llm-veridicality}(c) then gives $\inf_{\theta \in \Theta} L_D^*(\theta) > H(Y \mid C,W)$.
\end{proof}

Scaling can improve two different objects: (a)~realisability, in that larger classes $\Theta$ may realise finer partitions $p_{\theta,D}$; and (b)~ecology, in that broader data can increase $k_D$ by separating more pairs. Hence non-lossy behaviour is an empirical question about the pair $(\Theta, \mu_D)$, not a universal consequence of parameter count alone.

Representational capacity is not the only bottleneck. Even when the non-lossy regime is feasible and a fixed model achieves $\Delta_D(\theta)=0$, realized deployment loss can still remain above the entropy floor through a positive deployment decoding gap $\Gamma_D^{\mathcal{Q}_{\mathrm{dep}}}(\theta)$ from \defref{def:decoder-gap}. Chain-of-thought prompting and scratchpads are relevant on that axis: for fixed weights they can reduce the decoding gap by making an available distinction easier to exploit at readout time \citep{wei_etal_2022_cot,nye_etal_2021_scratchpads}, but they do not change the representational term $\Delta_D(\theta)$. Ecology injection or broader training data are needed when the distinction is absent from the frozen encoding itself. The present framework proves statements about the representational side of that divide; bounding the decoding gap for realistic transformer inference regimes remains open.

\section{The LLM Ecosystem as an Evolutionary System}\label{sec:evolution}

\subsection{Units of Selection}

The relevant level distinction is the same as in \citet{dallariva_2026}. SGD within the training run of a single model is \emph{developmental optimisation}, not the population process analysed by Price's equation or quasispecies theory. The relevant evolutionary entities are \emph{whole trained models and model lineages}: frozen artefacts that are copied, modified, deployed, retained, distilled into successors, or abandoned. Selection across such lineages is the population-level process.

\begin{table}[ht]
\centering
\begin{tabular}{@{}ll@{}}
\toprule
Ecological-veridicality framework & LLM ecosystem \\
\midrule
Organism & A trained model (full weights, frozen at deployment) \\
Population & The set of extant models and variants \\
Encoding $p$ & Induced world-state encoding $p_{\theta,D}$ \\
Fitness $\mathcal{F}(p)$ & Multi-task benchmark performance \\
Reproduction & Fine-tuning, distillation, next-generation training \\
Mutation & Architecture changes, data mix, RLHF \\
Horizontal transfer & Attention, MoE, RLHF spreading across labs \\
\bottomrule
\end{tabular}
\end{table}

\begin{definition}[Model-lineage population]
Fix a time horizon over which the deployment ecology $\mu$ remains approximately stationary. A \emph{model-lineage population} is a finite set of deployed or developmentally active lineages $\{\theta_1,\ldots,\theta_K\}$, where each lineage carries a frozen deployment encoding $p_{\theta_i,D}$ during evaluation, abbreviated $p_i$ below, may serve as a parent for successor lineages, and may generate descendants by checkpoint inheritance, distillation, fine-tuning, or retraining with modified architecture/data/objective.
\end{definition}

\begin{proposition}[Darwinian conditions hold at the inter-model level]\label{prop:darwinian-llm}
Suppose over a fixed horizon that:
\begin{enumerate}[label=(\alph*)]
\item descendant models inherit substantial structure from parent models (weights, architecture, tokenizer, training recipe, or dataset);
\item lineages vary in their induced encodings $p_\theta$ through such inherited modifications;
\item the probability that a lineage is copied, retained, fine-tuned, distilled, or used as the base for further training is increasing in its expected deployment success;
\item deployment success is evaluated on the performance of the whole trained model across the relevant task ecology.
\end{enumerate}
Then the model ecosystem instantiates heredity, variation, and differential reproduction at the level of whole trained models. In the sense relevant to the population theory of \citet{dallariva_2026}, it is therefore a Darwinian population of encodings.
\end{proposition}

\begin{proof}
Condition~(a) gives heredity, (b) gives variation, and (c)--(d) give differential reproduction on whole-model performance. The heritable trait under selection is the induced encoding $p_\theta$ carried by the lineage. SGD updates within a lineage are part of the developmental map from parent lineage to offspring lineage, not the population law itself.
\end{proof}

\subsection{Conditions for Importing the Ecological-Veridicality Population Dynamics}

\begin{proposition}[Selection dynamics across model lineages]\label{prop:llm-wf}
Consider a population of model lineages over a time window on which:
\begin{enumerate}[label=(\alph*)]
\item each active lineage $i$ carries a frozen deployment encoding $p_i$;
\item expected fitness is frequency-independent and depends on the encoding only through deployment performance, e.g.\
$\mathcal{F}(p_i) = C - \Delta_D(p_i)$ or any strictly decreasing transform of~$\Delta_D(p_i)$, where $\Delta_D(p_i):=L_D^*(p_i)-H(Y \mid C,W)$;
\item parent lineages are chosen with probability proportional to $\mathcal{F}(p_i)$;
\item offspring lineages inherit their parent's encoding up to a mutation kernel $Q$ on induced encodings (architecture changes, data changes, distillation noise, fine-tuning updates);
\item the mutation/reproduction process is Markovian on the induced encoding space over the chosen horizon.
\end{enumerate}
Then the population dynamics reduce to the same Wright--Fisher / replicator-mutator form analysed by \citet{dallariva_2026} on the induced encoding space $\mathcal{P}_\Theta$. Consequently, the same population model, together with its Price-equation and quasispecies consequences at the expectation/asymptotic level, applies conditionally to model populations, with the same caveat that convergence is only to the best mutation-accessible asymptotic regime unless stronger connectivity assumptions hold.
\end{proposition}

\begin{proof}
Under~(a)--(e), lineages are discrete heritable units carrying encodings, fitness is attached to those encodings, selection acts by weighted parent choice, and inherited modifications are represented by a mutation kernel $Q$. This is exactly the structure assumed by the population-level process model of \citet{dallariva_2026}, with organisms replaced by model lineages and perceptual encodings replaced by induced deployment encodings $p_i$. The conclusion is therefore a conditional structural reduction: once those assumptions hold, the same population theorems apply on the relabelled state space.
\end{proof}

If a common deployment decoder class $\mathcal{Q}_{\mathrm{dep}}$ is fixed across lineages and realized deployment performance rather than the oracle objective drives selection, the same formulation can instead use the realized excess
\[
\Delta_D(p_i) + \Gamma_D^{\mathcal{Q}_{\mathrm{dep}}}(p_i)
= L_D^{\mathcal{Q}_{\mathrm{dep}}}(p_i) - H(Y \mid C,W)
\]
in place of $\Delta_D(p_i)$. We retain the oracle form in the main text because the proved results in this paper characterize $\Delta_D$ directly, while $\Gamma_D^{\mathcal{Q}_{\mathrm{dep}}}$ is only structurally constrained.

\noindent\textbf{Consequences if these conditions hold.}
Over any window on which \propref{prop:llm-wf} is a good approximation, inter-model selection creates expectation-level pressure toward lower ecological excess loss and therefore toward more ecologically veridical induced encodings. The static theorems identify the target partition; the dynamic theorems of \citet{dallariva_2026} describe the conditional route by which a population of model lineages can move toward it. The conclusion remains conditional: convergence is only to the best mutation-accessible asymptotic regime.

\propref{prop:llm-wf} does not imply that SGD within a single training run obeys Price's equation or quasispecies theory. The proposition applies at the lineage level: once whole trained models are treated as the replicating entities, the inter-model process can satisfy the assumptions of the population theory. Some departures from the idealisation are \emph{benign}: performance-aligned reuse, distillation, architecture borrowing, directed engineering, and horizontal transfer across lineages can all accelerate search toward the same ecology-defined target without changing it, and mild frequency dependence need not destroy a local fixed-fitness approximation. The \emph{serious} failures are the target-changing ones: strong frequency dependence that reorders effective fitness by population composition, rapid non-stationarity of the deployment ecology, or engineering interventions that change the effective objective rather than merely the speed of search. The result is accordingly best read as a conditional framework for hypothesis generation and controlled experiments, not as a claim that the current production-model ecosystem literally satisfies the required assumptions. Those assumptions are more plausible in controlled microgpt populations than in commercial LLM markets.

\begin{figure}[t]
\centering
\includegraphics[width=0.92\linewidth]{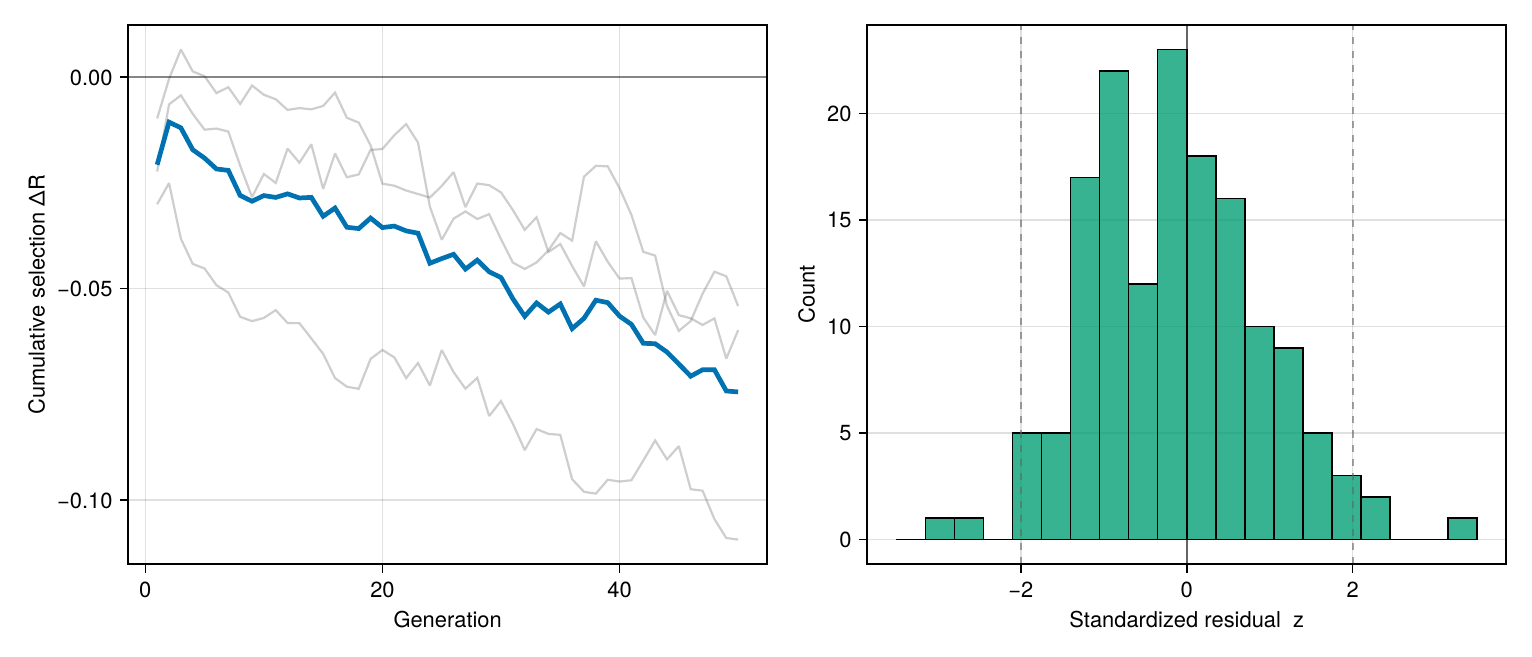}
\caption{Selection-stage diagnostics in the microgpt Wright--Fisher experiment. Left: cumulative selection-stage change in mean risk across generations. The downward trend shows a weak but persistent net selection pressure toward lower risk. Right: histogram of standardized residuals $z := (\Delta \bar R_{\mathrm{sel}} - \E[\Delta \bar R_{\mathrm{sel}} \mid \text{risk, fitness}])/\mathrm{sd}(\Delta \bar R_{\mathrm{sel}} \mid \text{risk, fitness})$. The residuals are centered and mostly fall inside $\pm 2$, with rms $z=1.05$ and exact 95\%-band coverage $0.96$, consistent with the Wright--Fisher conditional sampling law.}
\label{fig:exp3-dynamics}
\end{figure}

\subsection{Token and Evaluation Ecologies}\label{sec:two-ecology}

The static theorems above characterise optimality under a single ecology~$\mu$. Here we add the cases in which real LLM development is shaped by a second ecology beyond the base next-token objective, without replacing that one-ecology result. In the LLM setting, token-prediction training follows one ecology, while lineage retention and post-training may follow another; their interaction is naturally read as a Baldwin effect \citep{baldwin_1896,hinton_nowlan_1987}. The \emph{token ecology}~$\mu_{\mathrm{tok}}$ defines the single-run training target through next-token prediction. The \emph{evaluation ecology}~$\mu_{\mathrm{eval}}$ defines which model lineages are retained, invested in, fine-tuned, distilled into successors, and used as starting points for next-generation training through benchmarks, deployment, and user preferences.

These two ecologies have overlapping but generally non-nested separation sets. Many important world-state distinctions (mathematical validity, code correctness, long-range logical consistency) have only weak local next-token signatures, so $\sigmasq_{\mathrm{tok}}(w_1,w_2)$ may be small even when the evaluation ecology separates the pair strongly. Conversely, fine-grained orthographic patterns may be token-separated but evaluation-invisible.

The point is not that every global or structurally extended property requires a second ecology. Some such properties already have strong token-level signatures. \citet{gurnee_etal_2025_linebreaks}, for example, show that a next-token transformer can learn a low-dimensional ``character count manifold'' that tracks cumulative line length and supports line-break prediction from language modeling alone. Bracket balance can also be partly learned this way, as the experiments below illustrate. The two-ecology argument is needed for distinctions whose token-level signatures are too weak relative to simplicity pressure or competing variation in the training signal: not ``all nonlocal structure,'' but the gap cases for which $\sigmasq_{\mathrm{tok}}$ is small while $\sigmasq_{\mathrm{eval}}$ remains large.

To state that relationship precisely, we treat both ecologies as instances of the same formal object.

\begin{definition}[Generalized task ecology]\label{def:gen-ecology}
A \emph{generalized task ecology}~$\eta$ on a finite latent state space~$W$ with prior~$\pi$ consists of a probability measure over tasks~$t$, where each task has a query space~$Q_t$, a target space~$Y_t$, a query distribution~$D_t$, conditional target laws $P^t_w(\cdot\mid q)$ for each $w\in W$ and $q\in Q_t$, and a loss~$\ell_t$. For an encoding $p\colon W\to\X$ and a Bayes-optimal decoder family under~$\eta$, define the ecology-relative excess
\[
\Delta_\eta(p) := L^*_\eta(p) - L^*_\eta(\mathrm{id}_W),
\]
where $\mathrm{id}_W$ denotes the unreduced encoding.
\end{definition}

The token ecology instantiates this object with next-token prediction tasks under log loss. The evaluation ecology instantiates it with benchmark or deployment evaluations and their associated losses. Here we use the generalized object only to state separation sets and evaluation-relative excess; we do not invoke a full generalized analogue of \thmref{thm:ce-decomposition}. The pairwise separation functional under~$\eta$ is
\[
\sigmasq_\eta(w_1,w_2) := \E_{t\sim\eta}\,\E_{q\sim D_t}\bigl[d_t\bigl(P^t_{w_1}(\cdot\mid q),\,P^t_{w_2}(\cdot\mid q)\bigr)^2\bigr],
\]
where $d_t$ is a divergence on target laws that vanishes exactly on equality. Write $S_\eta := \{(w_1,w_2) : \sigmasq_\eta(w_1,w_2)>0\}$ for the separation set.

\begin{proposition}[Two-ecology scope]\label{prop:two-ecology-scope}
Let $\mu_{\mathrm{tok}}$ and $\mu_{\mathrm{eval}}$ be two ecologies on the same latent state space~$W$.
\begin{enumerate}[label=(\alph*)]
\item \textbf{Static scope.} If an encoding~$p$ satisfies $\Delta_{\mu_{\mathrm{tok}}}(p)=0$, then~$p$ preserves all and only the $\mu_{\mathrm{tok}}$-equivalence classes. Zero-excess token-ecology optimality constrains the partition of~$W$ only through~${\sim_{\mu_{\mathrm{tok}}}}$.

\item \textbf{Dynamic scope.} Suppose a model-lineage population satisfies the assumptions of \propref{prop:llm-wf}, and suppose expected lineage fitness has the form $\mathcal{F}(p)=\varphi(\Delta_{\mu_{\mathrm{eval}}}(p))$ for some strictly decreasing~$\varphi$. Then the same population dynamics apply with $\mu_{\mathrm{eval}}$ in place of~$\mu_{\mathrm{tok}}$: at the expectation level, selection pushes the population toward lower evaluation excess.

\item \textbf{Non-implication.} In general, $\Delta_{\mu_{\mathrm{tok}}}(p)=0$ does not imply $\Delta_{\mu_{\mathrm{eval}}}(p)=0$. A lineage process can be driven by evaluation-ecology fitness even on pairs for which $\mu_{\mathrm{tok}}$ gives only weak or vanishing separation.
\end{enumerate}
\end{proposition}

\begin{proof}
For part~(a), apply \thmref{thm:ce-decomposition}(c) to the token ecology~$\mu_{\mathrm{tok}}$: zero excess under that ecology is equivalent to preserving exactly the $\mu_{\mathrm{tok}}$-equivalence classes, so the static theorem constrains only that partition.

For part~(b), \propref{prop:llm-wf} requires only that expected fitness be a strictly decreasing function of the relevant ecology-relative excess. Replacing $\Delta_D$ there by $\Delta_{\mu_{\mathrm{eval}}}$ therefore leaves the structural reduction unchanged: parent choice is still weighted by fitness, offspring inherit encodings up to a mutation kernel, and the same Wright--Fisher / replicator-mutator conclusions apply on the induced-encoding space.

For part~(c), the two excess terms are tied to different separation structures. If $\mu_{\mathrm{eval}}$ separates a pair that $\mu_{\mathrm{tok}}$ leaves merged, then an encoding can have $\Delta_{\mu_{\mathrm{tok}}}(p)=0$ while still merging an evaluation-relevant distinction, which forces $\Delta_{\mu_{\mathrm{eval}}}(p)>0$. Hence token-optimality does not in general imply evaluation-optimality.
\end{proof}

This proposition makes explicit that the static optimality theorem and the evolutionary population theorem may be talking about different ecologies. Post-training provides a concrete mechanism for partially injecting the evaluation ecology into the token-prediction process. The next result formalises that mechanism.

\begin{proposition}[Ecology injection threshold]\label{prop:ecology-injection}
Let $\mu_0$ and~$\nu$ be two ecologies on the same latent state space~$W$, and for $\alpha\in[0,1]$ define the mixed ecology $\mu_\alpha := (1-\alpha)\mu_0 + \alpha\nu$. Then for every pair $(w_1,w_2)$:
\begin{enumerate}[label=(\alph*)]
\item \textbf{Exact interpolation.}
$\sigmasq_{\mu_\alpha}(w_1,w_2) = (1-\alpha)\,\sigmasq_{\mu_0}(w_1,w_2) + \alpha\,\sigmasq_\nu(w_1,w_2)$.

\item \textbf{Monotonicity.} If $\sigmasq_\nu(w_1,w_2) \ge \sigmasq_{\mu_0}(w_1,w_2)$, then $\sigmasq_{\mu_\alpha}(w_1,w_2)$ is nondecreasing in~$\alpha$; if the inequality is strict, it is strictly increasing.

\item \textbf{Threshold.} Fix an effective separation threshold~$\varepsilon>0$. If $\sigmasq_{\mu_0}(w_1,w_2)\le\varepsilon < \sigmasq_\nu(w_1,w_2)$, then the pair becomes effectively resolved under~$\mu_\alpha$ exactly when $\alpha > \alpha^*(w_1,w_2)$, where
\[
\alpha^*(w_1,w_2) := \frac{\varepsilon - \sigmasq_{\mu_0}(w_1,w_2)}{\sigmasq_\nu(w_1,w_2) - \sigmasq_{\mu_0}(w_1,w_2)}.
\]
\end{enumerate}
\end{proposition}

\begin{proof}
Part~(a): by linearity of expectation under the mixed measure,
\[
\sigmasq_{\mu_\alpha}(w_1,w_2) = (1-\alpha)\,\E_{t\sim\mu_0}[Z_t] + \alpha\,\E_{t\sim\nu}[Z_t],
\]
where $Z_t := \E_{q\sim D_t}[d_t(P^t_{w_1},P^t_{w_2})^2]$. Part~(b): the derivative with respect to~$\alpha$ is $\sigmasq_\nu - \sigmasq_{\mu_0}$. Part~(c): solve $(1-\alpha)\sigmasq_{\mu_0} + \alpha\sigmasq_\nu > \varepsilon$ for~$\alpha$.
\end{proof}

\begin{corollary}[Post-training refines token-ecology resolution]\label{cor:post-training}
Let $\mu_0$ be a base token ecology and~$\nu$ a post-training task family. Define $\mu_{\mathrm{tok}}^{(\alpha)} := (1-\alpha)\mu_0 + \alpha\nu$ for $\alpha\in[0,1]$.
\begin{enumerate}[label=(\alph*)]
\item For every $\alpha\in[0,1)$, the induced partition satisfies $[w]_{\mu_{\mathrm{tok}}^{(\alpha)}} \subseteq [w]_{\mu_0}$: post-training can split existing equivalence classes but cannot coarsen them.

\item If $(w_1,w_2)$ is a gap pair with $\sigmasq_{\mu_0}(w_1,w_2)\le\varepsilon$ and $\sigmasq_\nu(w_1,w_2)>\varepsilon$, then for every $\alpha>\alpha^*(w_1,w_2)$ from \propref{prop:ecology-injection}, the pair is resolved under~$\mu_{\mathrm{tok}}^{(\alpha)}$.

\item The rescued set $R_\alpha := \{(w_1,w_2)\in G_\varepsilon : \sigmasq_{\mu_{\mathrm{tok}}^{(\alpha)}}(w_1,w_2)>\varepsilon\}$ is nondecreasing in~$\alpha$ whenever $\sigmasq_\nu\ge\sigmasq_{\mu_0}$ pairwise on~$G_\varepsilon$.
\end{enumerate}
\end{corollary}

\begin{proof}
For~(a), if $\sigmasq_{\mu_0}(w_1,w_2)>0$ and $\alpha<1$, then \propref{prop:ecology-injection}(a) gives $\sigmasq_{\mu_\alpha}(w_1,w_2)\ge(1-\alpha)\sigmasq_{\mu_0}(w_1,w_2)>0$, so every pair separated by~$\mu_0$ remains separated. For~(b), apply \propref{prop:ecology-injection}(c). For~(c), each pairwise score is nondecreasing in~$\alpha$ by \propref{prop:ecology-injection}(b), so once a pair enters~$R_\alpha$ it remains for all larger~$\alpha$.
\end{proof}

The two-ecology picture refines the failure predictions of \cref{sec:conclusion}. Models should fail on distinctions where \emph{both} $\sigmasq_{\mathrm{tok}}\approx 0$ and $\sigmasq_{\mathrm{eval}}\approx 0$. On distinctions where $\sigmasq_{\mathrm{tok}}\approx 0$ but $\sigmasq_{\mathrm{eval}}\gg 0$, the evolutionary dynamics provide pressure through lineage selection, and post-training injects the evaluation signal into the token-prediction process with an explicit threshold. The rate of improvement on such gap pairs is controlled by the efficiency of ecology injection.

\paragraph{Model-organism checks.}
Two microgpt experiments test the two-ecology mechanism on bracket balance in real Lisp source code (from \emph{Practical Common Lisp}). Both use the same design: a recipe trait $\alpha\in[0,1]$ controls ecology injection, a static sweep measures the effect of varying~$\alpha$, and a Wright--Fisher population selects on evaluation fitness. The experiments are named by what the evaluation ecology tests, not by what the model is trained to do (which is always next-token prediction).

In the \emph{balance checking} task, the world states are balanced versus unbalanced Lisp chunks, with a summary token appended to indicate bracket balance. The token ecology trains on the chunks without the summary; post-training at level~$\alpha$ mixes in the labeled version. The underlying global property, bracket nesting, is structural and capacity-limited: on held-out evaluation, summary cross-entropy falls gradually from $25.6$ at $\alpha=0$ to $0.19$ at $\alpha=1.0$, while selection raises $\bar{\alpha}$ from $0.46$ to $0.92$. This experiment demonstrates ecology injection on a genuinely hard structural task, but the evaluation signal leaks into training through the summary token itself.

The \emph{minimal code validation} task removes that leakage. The recipe trait~$\alpha$ controls only the fraction of bracket-containing Lisp code in the next-token training corpus; at $\alpha=0$ the model trains on the same code with brackets scrubbed out. No balance labels or summary tokens appear during training. Evaluation measures the held-out NLL gap between balanced and bracket-permuted chunks: a model that has learned bracket structure from real Lisp should find valid code more predictable than structurally scrambled code. At $\alpha=0$ the model is blind to bracket balance (discrimination $-0.002$); at $\alpha=0.1$ discrimination rises to $0.46$, and it increases steadily to $0.78$ at $\alpha=1.0$. The transfer is indirect: bracket exposure during next-token prediction develops sensitivity to a structural distinction that is never directly supervised. Population selection shows a noisy but clearly upward trajectory (from $\bar{\alpha}=0.47$ to $0.89$ over 25~generations), with $\bar{\alpha}_{\mathrm{eval}} \ge \bar{\alpha}$ in nearly every generation, leaving the final population concentrated on bracket-rich recipes. \Cref{fig:exp67-two-ecology} summarizes the static and population-level patterns.

\begin{figure}[!htbp]
\centering
\includegraphics[width=0.94\linewidth]{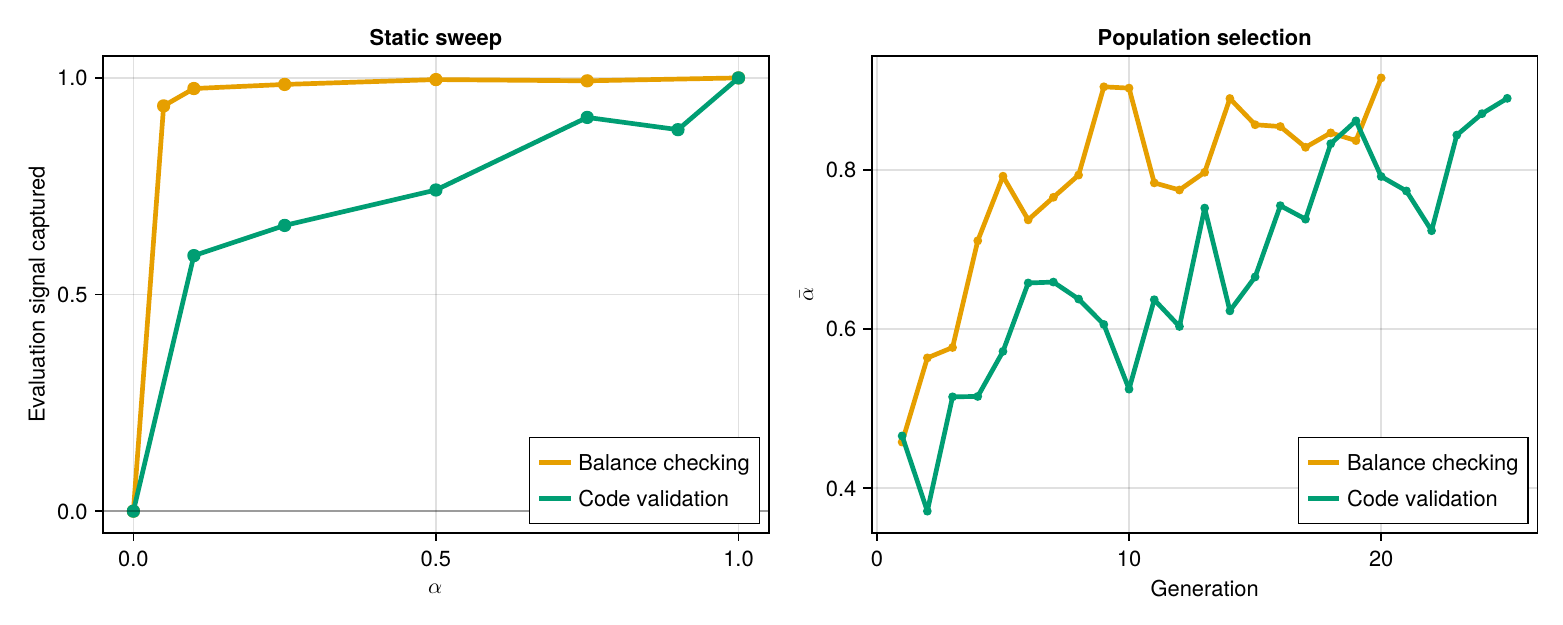}
\caption{Neural validation of the two-ecology mechanism on bracket balance in Lisp source code. Left: static sweep showing the fraction of maximum evaluation signal captured as a function of the recipe trait~$\alpha$. Both tasks start at zero signal when $\alpha=0$. Balance checking (direct supervision via summary token) saturates quickly; code validation (no balance labels, held-out NLL discrimination only) rises more steadily, confirming that the transfer is indirect. Right: population selection drives~$\bar{\alpha}$ upward in both tasks, concentrating the recipe distribution on bracket-rich training.}
\label{fig:exp67-two-ecology}
\end{figure}
\FloatBarrier

\paragraph{Non-stationarity and directed variation.}
Real LLM ``mutations'' are directed (Lamarckian): engineers observe failures and design improvements. Architectural innovations, training practices, and weight sharing spread across labs by horizontal transfer. These features can all accelerate convergence toward the same ecology-defined target without breaking the framework, so long as they do not change the effective objective. The task ecology~$\mu$ does shift over time (new benchmarks, new user demands), creating Red Queen dynamics where the population must track a moving optimum. Within any window of approximate stationarity, however, the static theorems identify the target partition and the population dynamics describe the conditional path toward it. With that dynamic bridge in place, the next question is what target such pressure selects when ecological veridicality is achievable.

\section{Minimum-Complexity Ecological Veridicality}\label{sec:mincomplex}

The static theorem identifies when zero excess is achievable, but not which zero-excess encoding should be preferred when several are available. In this section, we add a simplicity refinement on top of that static result: among all ecologically veridical encodings, which one preserves only the task-relevant distinctions and no more? The results below are stated for a generic ecology~$\mu$; they apply equally to the token ecology~$\mu_{\mathrm{tok}}$, the evaluation ecology~$\mu_{\mathrm{eval}}$, or any mixture.

\subsection{The Minimum-Complexity Theorem}

\begin{definition}[Representational complexity]
For an encoding $p \colon W \to X$ with prior $\pi$, the representational complexity is $I(p) = I(W; p(W)) = H(p(W))$, since $p$ is deterministic.
\end{definition}

\begin{theorem}[Minimum-complexity veridicality]\label{thm:min-complexity}
Among all encodings with
\[
L_D^*(p) = H(Y \mid C,W)
\]
(equivalently: among all ecologically veridical encodings under the training ecology):
\begin{enumerate}[label=(\alph*)]
\item The minimum representational complexity is:
\[
I^*(\mu) = H(W/{\sim_\mu}) = -\sum_{[w] \in W/{\sim_\mu}} \pi([w]) \log \pi([w]),
\]
where
\[
\pi([w]) := \sum_{u \in [w]} \pi(u)
\]
is the total prior mass of the $\mu$-equivalence class $[w]$;

\item This minimum is achieved by encodings whose partition is exactly $W/{\sim_\mu}$, no finer and no coarser.

\item Any strictly finer encoding (e.g.\ fully veridical when some $|[w]| > 1$) has $I(p) > H(W/{\sim_\mu})$. For the fully veridical encoding, $I(p) = H(W)$, so the maximal excess complexity is $H(W) - H(W/{\sim_\mu}) = H(W \mid W/{\sim_\mu})$, the within-class entropy.
\end{enumerate}
\end{theorem}

\begin{proof}
By \thmref{thm:ce-decomposition}(c), attaining $L_D^*(p)=H(Y \mid C,W)$ is equivalent to each cell containing only $\mu$-equivalent states. The partition induced by $p$ must therefore refine the quotient partition $W/{\sim_\mu}$. Let $\mathcal{Q} := W/{\sim_\mu}$ and let $\mathcal{P} := p(W)$. Because $\mathcal{P}$ refines $\mathcal{Q}$, the grouping identity gives
\[
H(\mathcal{P}) = H(\mathcal{Q}) + H(\mathcal{P} \mid \mathcal{Q}) \ge H(\mathcal{Q}),
\]
with equality iff $H(\mathcal{P} \mid \mathcal{Q})=0$, i.e.\ iff $\mathcal{P}=\mathcal{Q}$. This proves~(a): the minimum possible representational complexity among zero-excess encodings is $H(\mathcal{Q})=H(W/{\sim_\mu})$. It also proves~(b): the minimizers are exactly the encodings whose partition is $W/{\sim_\mu}$ itself, neither finer nor coarser. For~(c), any strictly finer zero-excess encoding has $H(\mathcal{P} \mid \mathcal{Q})>0$, hence $I(p)=H(\mathcal{P})>H(\mathcal{Q})=H(W/{\sim_\mu})$. The fully veridical encoding corresponds to the identity partition on~$W$, so its complexity is $H(W)$, and the excess complexity relative to the minimum is
\[
H(W) - H(W/{\sim_\mu}) = H(W \mid W/{\sim_\mu}).
\]
\end{proof}

\begin{interpretation}
The minimum-complexity ecologically veridical encoding carries exactly the task-relevant information and nothing else. This gives a precise entropy-based benchmark for what a simplicity preference would have to select: among all zero-excess representations, the coarsest partition compatible with the ecology. Any extra resolution within $\mu$-equivalence classes carries additional information cost without improving Bayes-optimal token loss.
\end{interpretation}

\begin{corollary}[Topological convergence of optima]\label{cor:topological-convergence}
If two models $\theta_1, \theta_2$ both attain the training optimum
\[
L_D^*(\theta_i) = H(Y \mid C,W)
\]
and both have minimum representational complexity under the same $\mu_D$, then $p_{\theta_1,D}$ and $p_{\theta_2,D}$ induce exactly the same partition $W/{\sim_{\mu_D}}$. Consequently they agree on the zero/nonzero separation pattern, and any kernel built from the $k_D := |W/{\sim_{\mu_D}}|$ distinct class codes has rank at most $k_D-1$, with equality only under non-degenerate geometry.
\end{corollary}

\begin{proof}
By \thmref{thm:min-complexity}(b), every minimum-complexity training-optimal encoding induces the quotient partition $W/{\sim_{\mu_D}}$ and no finer one. Therefore $p_{\theta_1,D}$ and $p_{\theta_2,D}$ identify exactly the same world-state pairs, namely the $\mu_D$-equivalent pairs, so they agree on the full zero/nonzero separation pattern.

For the rank statement, both encodings realize exactly $k_D:=|W/{\sim_{\mu_D}}|$ distinct class codes. After centering, those class representatives lie in an affine subspace of dimension at most $k_D-1$, because the centered representatives sum to zero. Any centered Gram matrix built from them therefore has rank at most $k_D-1$, with equality only when the $k_D$ class representatives are in affine general position.
\end{proof}

\subsection{The Rate-Distortion Curve}

The minimum-complexity theorem identifies the first zero-excess point. It is also useful to phrase the same fact as a rate-distortion statement: how much representational complexity is required before zero excess becomes achievable at all? The next corollary makes that threshold explicit.

\begin{corollary}[Rate-distortion characterisation]
Define the excess-loss distortion
\[
R(p) := L_D^*(p) - H(Y \mid C,W),
\]
and the induced rate-distortion function
\[
R(I) := \min_{p:\, I(W;p(W)) \le I} R(p).
\]
Then $R(I) = 0$ for $I \ge I^*(\mu)$ and $R(I) > 0$ for $I < I^*(\mu)$. The critical rate $I^*(\mu)$ is the phase transition point from strictly positive excess loss to zero excess loss.
\end{corollary}

\begin{proof}
By \thmref{thm:min-complexity}, zero excess is achievable exactly for encodings whose complexity is at least the minimum zero-excess complexity~$I^*(\mu)$. Hence if $I\ge I^*(\mu)$, the feasible set in the definition of $R(I)$ contains a zero-excess encoding, so $R(I)=0$. If $I<I^*(\mu)$, then no encoding with complexity at most~$I$ can attain zero excess, again by \thmref{thm:min-complexity}; therefore every feasible encoding has strictly positive distortion, and so does their minimum.
\end{proof}

$I^*(\mu)$ is determined by the task ecology, not the model. Scaling the model does not change $I^*(\mu)$; scaling the data changes $\mu$ and hence $I^*(\mu)$. If optimisation has a simplicity preference, $I^*(\mu)$ is the lower bound it would favour among zero-excess encodings. Whether SGD exhibits such a preference strongly enough to drive $p_{\theta,D}$ near this bound is an additional empirical and theoretical question.

\subsection{Local Split Criterion under Simplicity Pressure}

The minimum-complexity result is global: it compares entire zero-excess encodings. To derive concrete failure predictions, we also want a local criterion saying when a distinction is worth preserving under an explicit simplicity pressure. The next setup isolates a single candidate split and computes the exact gain from resolving it.

\begin{definition}[Complexity-regularized token objective]\label{def:regularized-objective}
For $\beta \ge 0$ and encoding $p \colon W \to X$, define
\[
J_{D,\beta}(p) := L_D^*(p) + \beta\, I(W; p(W))
            = L_D^*(p) + \beta\, H(p(W)).
\]
This objective is not identified with the exact SGD objective. It serves instead as an explicit model of a simplicity pressure that trades predictive performance against representational complexity.
\end{definition}

\begin{definition}[One-cell refinement]\label{def:one-cell-refinement}
Let $p$ be an encoding and let $S \subseteq W$ be one of its cells with $\pi_S := \sum_{w \in S}\pi(w) > 0$. Partition $S$ into two non-empty subcells $A$ and $B$, write
\[
\pi_A := \sum_{w \in A}\pi(w),
\qquad
\pi_B := \sum_{w \in B}\pi(w),
\qquad
\lambda := \pi_A/\pi_S,
\]
and let $p^{A|B}$ be the refinement obtained by replacing cell $S$ with the two cells $A$ and $B$ and leaving all other cells unchanged.

For each context $c$, define the subcell-average next-token distributions
\[
\bar P_A(\cdot \mid c)
:= \sum_{w \in A} \frac{\pi(w)}{\pi_A}\, P_w(\cdot \mid c),
\qquad
\bar P_B(\cdot \mid c)
:= \sum_{w \in B} \frac{\pi(w)}{\pi_B}\, P_w(\cdot \mid c).
\]
\end{definition}

\begin{theorem}[Split-versus-merge threshold]\label{thm:split-threshold}
In the setup of \defref{def:one-cell-refinement},
\[
J_{D,\beta}(p) - J_{D,\beta}(p^{A|B})
= \pi_S \Bigl(
    \E_{c \sim D_C}\bigl[
      \JS_{\lambda}\bigl(\bar P_A(\cdot \mid c),\, \bar P_B(\cdot \mid c)\bigr)
    \bigr]
    - \beta\, h(\lambda)
  \Bigr),
\]
where
\[
h(\lambda) := -\lambda \log \lambda - (1-\lambda)\log(1-\lambda)
\]
is the binary entropy.

Consequently:
\begin{enumerate}[label=(\alph*)]
\item the refinement $p^{A|B}$ is preferred to $p$ under $J_{D,\beta}$ iff
\[
\E_{c \sim D_C}\bigl[
  \JS_{\lambda}\bigl(\bar P_A(\cdot \mid c),\, \bar P_B(\cdot \mid c)\bigr)
\bigr]
> \beta\, h(\lambda);
\]

\item the merge is preferred iff the opposite inequality holds;

\item when $\beta > 0$, distinctions with sufficiently small predictive Jensen--Shannon gain are optimally merged.
\end{enumerate}
\end{theorem}

\begin{proof}
Let $X := p(W)$ and $X' := p^{A|B}(W)$. Since $X$ is a deterministic function of $X'$, the loss difference is
\[
L_D^*(p) - L_D^*(p^{A|B})
= H(Y \mid C, X) - H(Y \mid C, X')
= I(Y; X' \mid C, X).
\]
Only the split cell contributes. More explicitly, if $X=x$ with $x \neq S$, then $X'=x$ deterministically as well, so $I(Y;X' \mid C, X=x)=0$. Outside cell $S$, $X'$ therefore carries no extra information beyond $X$. On the original cell $S$, the refinement amounts to a binary label $Z \in \{A,B\}$ with $P(Z=A \mid X=S)=\lambda$ and $P(Z=B \mid X=S)=1-\lambda$. Therefore
\[
I(Y; X' \mid C, X)
= \pi_S\, I(Y; Z \mid C, X=S)
= \pi_S\, \E_{c \sim D_C}\bigl[
	    \JS_{\lambda}\bigl(\bar P_A(\cdot \mid c),\, \bar P_B(\cdot \mid c)\bigr)
	  \bigr].
\]
For the complexity term, splitting one cell of mass $\pi_S$ into masses $\pi_A$ and $\pi_B$ increases entropy by the grouping identity:
\[
H(X') - H(X)
= -\pi_A \log \pi_A - \pi_B \log \pi_B + \pi_S \log \pi_S
= \pi_S\, h(\lambda).
\]
Subtracting $\beta(H(X')-H(X))$ from the loss improvement gives the stated formula for $J_{D,\beta}(p) - J_{D,\beta}(p^{A|B})$. Parts~(a)--(c) are immediate.
\end{proof}

\begin{interpretation}
The quantity
\[
\Delta_{\mathrm{pred}}(A,B)
:= \E_{c \sim D_C}\bigl[
     \JS_{\lambda}\bigl(\bar P_A(\cdot \mid c),\, \bar P_B(\cdot \mid c)\bigr)
   \bigr]
\]
is the predictive value of resolving the distinction $A$ versus $B$ under the ecology in question. The complexity cost of doing so is the binary entropy term $h(\lambda)$. Under the explicit encoding-level objective $J_{D,\beta}$, distinctions whose predictive gain is too small relative to that cost are locally preferred merge candidates. What is proved here is a local comparison between $p$ and one refinement $p^{A|B}$ under that explicit objective; this theorem by itself does not identify the exact SGD objective, nor does it imply that ordinary parameter-space regularizers such as weight decay generate a globally monotone merge path.

\citet{elhage_etal_2022_superposition} suggest a plausible implementation-level picture for this threshold in actual transformers: under capacity pressure, weak features need not disappear discretely, but can be stored in superposition, with noisier downstream readout than strongly useful features. We use that only as a mechanistic interpretation of how weak distinctions may become fragile under simplicity pressure, not as a derivation of \thmref{thm:split-threshold}.

The present theorem is purely representational: it favors lower-entropy partitions among zero-excess encodings. Under restricted deployment inference classes there may be a second, computational analogue of simplicity pressure, favoring encodings whose preserved distinctions are easier to exploit and therefore induce smaller decoding gaps $\Gamma_D^{\mathcal{Q}_{\mathrm{dep}}}$. A joint theory of representational and computational simplicity remains open.

The split-threshold criterion applies to any ecology, not only the token ecology. In the two-ecology setting of \cref{sec:two-ecology}, a distinction that is a merge candidate under $\mu_{\mathrm{tok}}$ alone (because $\Delta_{\mathrm{pred}}(A,B)$ is small under token prediction) may nevertheless be preserved if ecology injection raises the effective separation above the threshold: once $\alpha > \alpha^*(w_1,w_2)$ from \propref{prop:ecology-injection}, the injected ecology contributes enough predictive gain that simplicity pressure no longer favours the merge.
\end{interpretation}

\begin{definition}[One-step partition neighborhood]\label{def:partition-neighborhood}
Identify an encoding $p$ with its induced partition of $W$ into non-empty cells. Define:
\begin{align*}
N_{\mathrm{split}}(p)
&:= \{p' : p' \text{ is obtained from } p \text{ by splitting one cell into two non-empty subcells}\},\\
N_{\mathrm{merge}}(p)
&:= \{p' : p' \text{ is obtained from } p \text{ by merging two distinct cells}\},\\
N(p) &:= N_{\mathrm{split}}(p) \cup N_{\mathrm{merge}}(p).
\end{align*}
We call $p$ a \emph{local minimum} of $J_{D,\beta}$ on the partition lattice if
\[
J_{D,\beta}(p) \le J_{D,\beta}(p')
\qquad
\text{for every } p' \in N(p).
\]
\end{definition}

\begin{proposition}[Local minima on the partition lattice]\label{prop:partition-local-min}
An encoding $p$ is a local minimum of $J_{D,\beta}$ if and only if both of the following conditions hold:
\begin{enumerate}[label=(\alph*)]
\item \textbf{Split stability.} For every cell $S$ of $p$ and every non-trivial bipartition $S=A \sqcup B$ with $\lambda=\pi(A)/\pi(S)$,
\[
\E_{c \sim D_C}\bigl[
  \JS_{\lambda}\bigl(\bar P_A(\cdot \mid c),\, \bar P_B(\cdot \mid c)\bigr)
\bigr]
\le \beta\, h(\lambda).
\]

\item \textbf{Merge stability.} For every pair of distinct cells $C_1,C_2$ of $p$, with $\pi_{C_1\cup C_2}=\pi(C_1)+\pi(C_2)$ and $\lambda=\pi(C_1)/\pi_{C_1\cup C_2}$,
\[
\E_{c \sim D_C}\bigl[
  \JS_{\lambda}\bigl(\bar P_{C_1}(\cdot \mid c),\, \bar P_{C_2}(\cdot \mid c)\bigr)
\bigr]
\ge \beta\, h(\lambda).
\]
\end{enumerate}
\end{proposition}

\begin{proof}
By \thmref{thm:split-threshold}(a), a one-step split lowers $J_{D,\beta}$ exactly when the corresponding Jensen--Shannon gain exceeds $\beta h(\lambda)$. Hence condition~(a) is equivalent to $J_{D,\beta}(p)\le J_{D,\beta}(p')$ for every $p' \in N_{\mathrm{split}}(p)$.

For a one-step merge of two cells $C_1,C_2$, let $p'$ denote the merged partition and view $p$ as the refinement of $p'$ obtained by splitting $C_1 \cup C_2$ back into $C_1$ and $C_2$. Applying \thmref{thm:split-threshold}(b) to that split shows that the merge lowers $J_{D,\beta}$ exactly when the same Jensen--Shannon gain is $< \beta h(\lambda)$. Thus condition~(b) is equivalent to $J_{D,\beta}(p)\le J_{D,\beta}(p')$ for every $p' \in N_{\mathrm{merge}}(p)$.

Combining the two equivalences and using $N(p)=N_{\mathrm{split}}(p)\cup N_{\mathrm{merge}}(p)$ proves the claim.
\end{proof}

\begin{corollary}[Local stability of the minimum-complexity veridical partition]\label{cor:veridical-local-min}
Let $p^\star$ be a minimum-complexity zero-excess encoding, so that its partition is exactly $W/{\sim_{\mu_D}}$ by \thmref{thm:min-complexity}. Then $p^\star$ is split-stable for every $\beta \ge 0$, and it is a local minimum of $J_{D,\beta}$ if and only if
\[
\beta \le \beta_{\min},
\]
where
\[
\beta_{\min}
:=
\min_{\substack{C_1,C_2 \in W/{\sim_{\mu_D}}\\ C_1 \neq C_2}}
\frac{
\E_{c \sim D_C}\bigl[
  \JS_{\lambda}\bigl(\bar P_{C_1}(\cdot \mid c),\, \bar P_{C_2}(\cdot \mid c)\bigr)
\bigr]
}{
h(\lambda)
},
\qquad
\lambda := \frac{\pi(C_1)}{\pi(C_1)+\pi(C_2)}.
\]
\end{corollary}

\begin{proof}
If $A,B$ lie inside a single $\mu_D$-equivalence class, then $P_w(\cdot \mid c)$ is the same for all $w \in A \cup B$ for $D_C$-almost every $c$. Hence $\bar P_A(\cdot \mid c)=\bar P_B(\cdot \mid c)$ almost everywhere and the split-gain term in \propref{prop:partition-local-min}(a) is zero. So every within-class split is neutral or disfavored, which proves split stability.

For merges between distinct $\mu_D$-classes, \propref{prop:partition-local-min}(b) shows that local stability is equivalent to requiring the Jensen--Shannon gain of every class pair to be at least $\beta h(\lambda)$. Taking the minimum over class pairs gives the threshold $\beta_{\min}$.
\end{proof}

\begin{remark}[Limits of the local criterion]\label{rem:beta-path-caution}
At $\beta=0$, every zero-excess partition is a local minimum of $J_{D,0}=L_D^*$, and \thmref{thm:min-complexity} selects $p^\star$ as the coarsest such partition. As $\beta$ increases past $\beta_{\min}$, \corref{cor:veridical-local-min} identifies exactly which distinction first becomes locally unstable: the class pair with the smallest Jensen--Shannon-gain-to-entropy-cost ratio.

Beyond that first threshold, however, the local criterion must be recomputed on the updated partition. Once two cells merge, both the Jensen--Shannon gains and the weights $\lambda$ change, so later transitions are determined by the current partition rather than by the original pairwise ordering alone. The theorem therefore gives an exact characterization of one-step local stability, but not a complete global merge path through partition space or parameter space.
\end{remark}

\begin{figure}[t]
\centering
\includegraphics[width=0.78\linewidth]{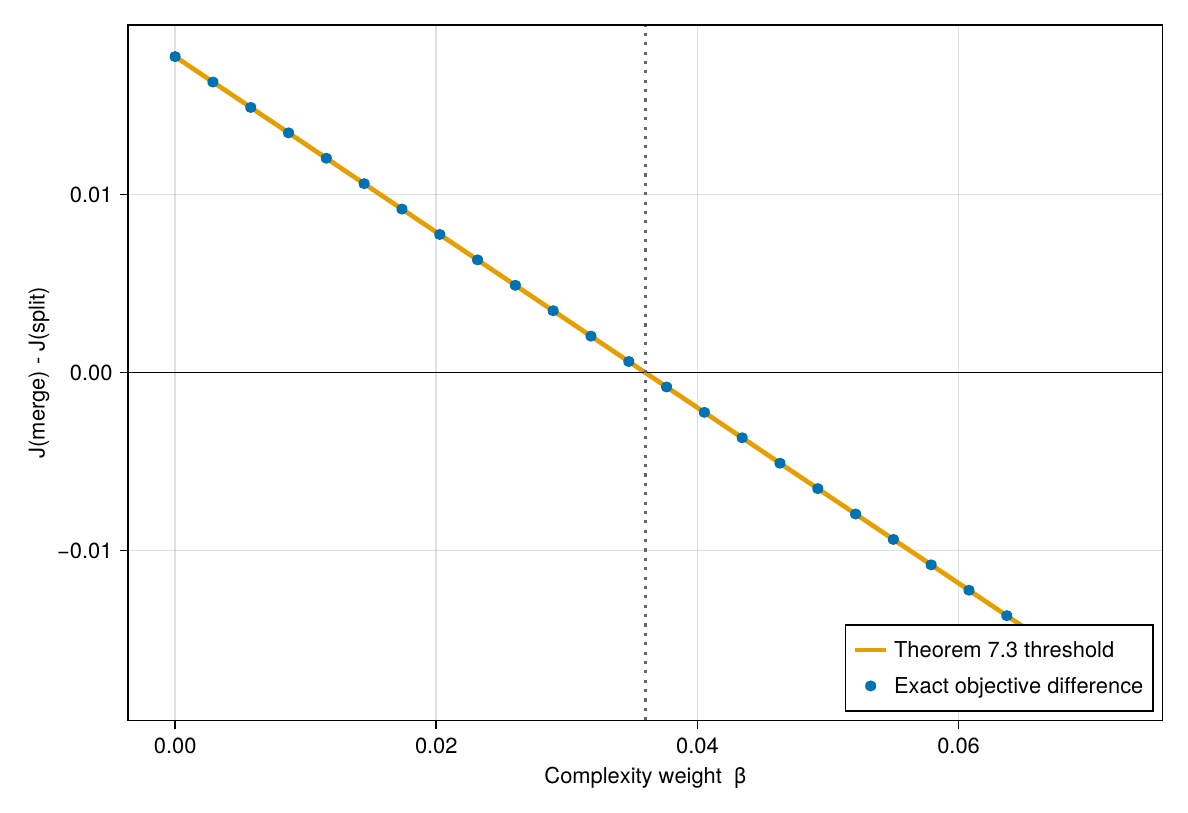}
\caption{Exact finite-ecology calibration of \thmref{thm:split-threshold}. The plotted quantity is the objective difference $J_{D,\beta}(p_{\mathrm{merge}})-J_{D,\beta}(p_{\mathrm{split}})$ for an illustrative local split. The crossing occurs at the theorem's threshold $\beta^\ast$: below it the split is preferred, above it the merge is preferred.}
\label{fig:exp0-split-threshold}
\end{figure}

\begin{figure}[t]
\centering
\includegraphics[width=\linewidth]{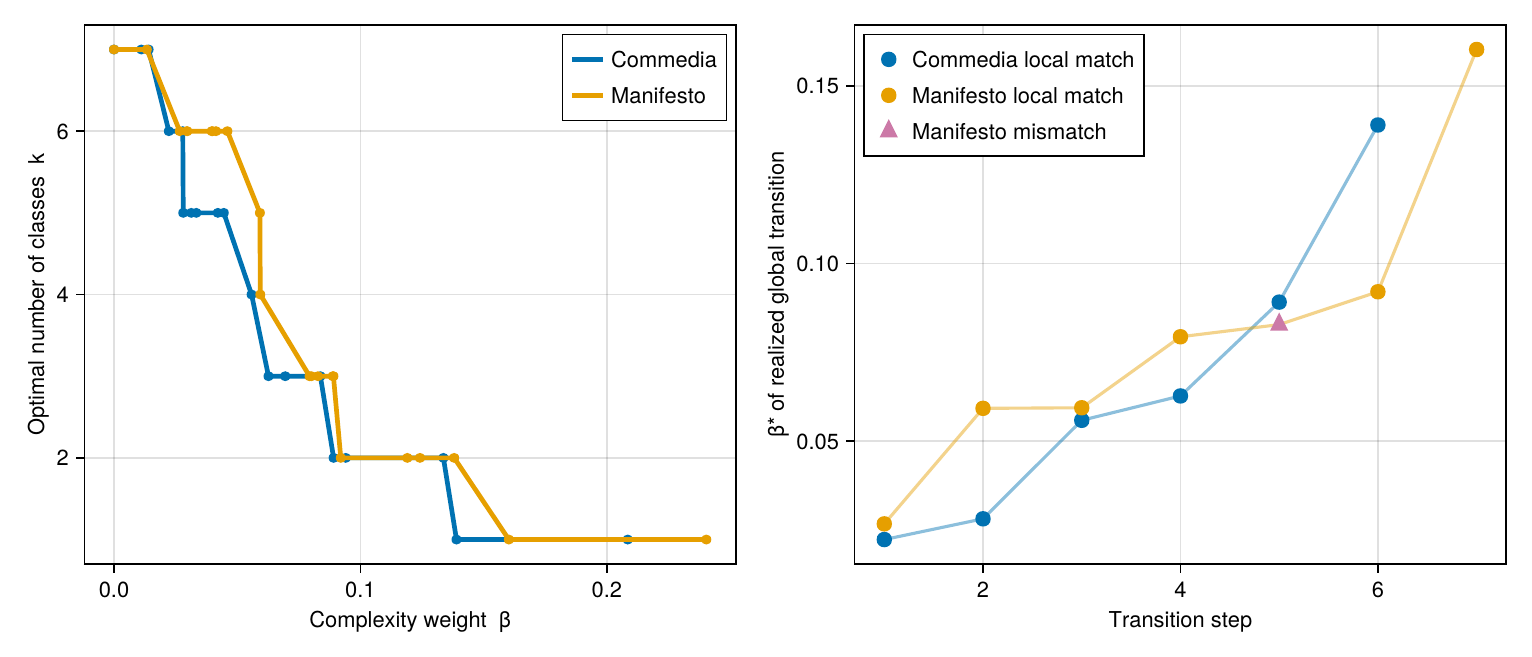}
\caption{Exact corpus-induced test of \thmref{thm:split-threshold}. Left: the exact global optimum path under $J_{D,\beta}$ for the \emph{Commedia} and \emph{Manifesto} ecologies as $\beta$ increases. Right: the realized global transitions are compared step-by-step to the theorem's local split-threshold prediction. The alignment is exact on \emph{Commedia} and has a single nonlocal deviation on \emph{Manifesto}.}
\label{fig:exp2-split-merge}
\end{figure}
 
\section{Predictions, Limits, and Conclusion}\label{sec:conclusion}

Together, the decomposition theorem, the minimum-complexity result, and the two-ecology framework identify where representational failure should occur. The logic requires no new propositions beyond those already proved. \Cref{app:off-ecology} adds quantitative lower bounds on off-ecology excess and a constructive non-identifiability witness.

\paragraph{Merged distinctions incur excess.}
If an encoding merges a pair $(w_1,w_2)$ that a probe ecology separates, \thmref{thm:ce-decomposition}(b) immediately gives positive excess under that ecology. By \thmref{thm:min-complexity}, a minimum-complexity zero-excess encoding for ecology~$\mu$ merges exactly the $\mu$-equivalent pairs. Any probe ecology that refines~$\mu$ therefore exposes positive excess on the newly separated pairs.

\paragraph{Simplicity pressure sheds low-gain distinctions first.}
Under the regularized objective~$J_{D,\beta}$, \thmref{thm:split-threshold} shows that distinctions whose predictive Jensen--Shannon gain is smaller than $\beta\,h(\lambda)$ are locally preferred merge candidates.
The pairs with the smallest gain-to-cost ratio are the first to become unstable as $\beta$ increases~(\corref{cor:veridical-local-min}).

\paragraph{Token and evaluation ecologies may disagree.}
The two-ecology framework (\cref{sec:two-ecology}) identifies the gap set: pairs where $\sigmasq_{\mathrm{tok}}(w_1,w_2)\approx 0$ but $\sigmasq_{\mathrm{eval}}(w_1,w_2)\gg 0$. On such pairs the token ecology provides little pressure to preserve the distinction, but the evaluation ecology rewards it. Ecology injection (\propref{prop:ecology-injection}) can rescue gap pairs, with the required injection level given by the explicit threshold~$\alpha^*$.

\begin{figure}[t]
\centering
\includegraphics[width=\linewidth]{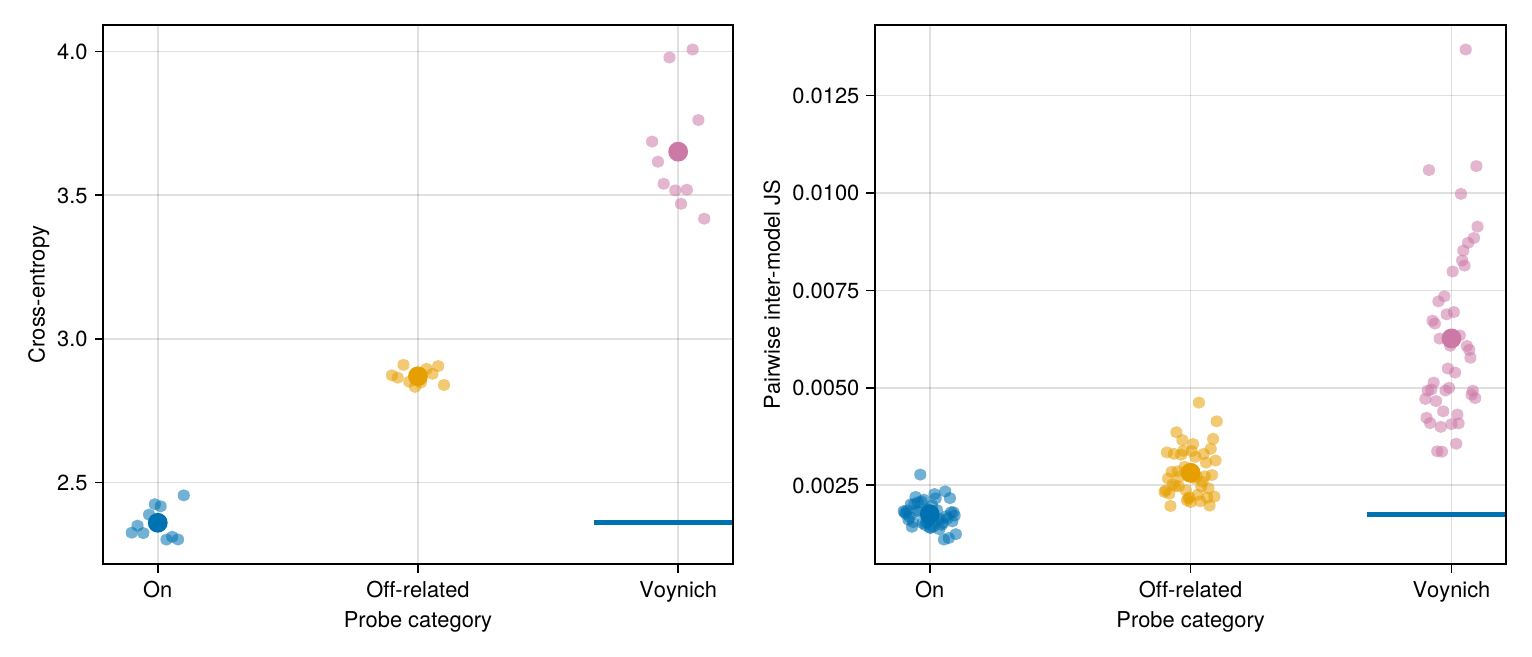}
\caption{Off-ecology failure in the microgpt model organism. Left: per-model cross-entropy rises from the training ecology (English, French, German) to related unseen languages (Italian, Finnish) and rises again on Voynich. Right: pairwise inter-model Jensen--Shannon disagreement shows the same ordering. This is the empirical signature predicted by the decomposition theorem: off-ecology probes incur larger excess and leave more room for divergent model behaviour.}
\label{fig:exp4-off-ecology}
\end{figure}

\paragraph{Predictions for production models.}
The model-organism experiments validate the framework in a regime where every quantity is observable. For production-scale models, the same logic yields only proxy-level predictions: holding deployment query type fixed, error should be highest on distinctions with low predictive split gain; models trained on comparable ecologies should agree on strongly separated distinctions and diverge on weakly separated or off-ecology ones; adding a modality should expand~$\mu$ and resolve previously fused equivalence classes (\propref{prop:ecology-expansion}); and a generalist whose encoding achieves zero excess on a unified ecology should match or outperform specialists on each sub-ecology (\thmref{thm:gen-vs-spec}).

\paragraph{Why the model-organism approach matters.}
The framework matters scientifically before it matters engineering-wise. It lets us ask what representational pressure autoregressive training, simplicity bias, and inter-model selection create in language-model populations, and test those claims where the relevant quantities are observable rather than hidden. The resulting predictions for larger systems are therefore conditional and comparative, not direct measurements.

\paragraph{Representational geometry.}
The framework determines which distinctions optimal representations must preserve (their \emph{topology}) but not how far apart the distinguished states lie in representation space (their \emph{geometry}). The ecology induces a canonical Hilbertian target geometry through the task-distance kernel~$K_\mu$ (\cref{app:geometry}), but our theorems propagate that geometry to learned encodings only in the Gaussian-linear case (\thmref{thm:geometry-gaussian}). Mechanistic work provides suggestive empirical analogues without closing that gap: \citet{gurnee_etal_2025_linebreaks} exhibit low-dimensional manifolds tracking structural task variables in a next-token transformer, while \citet{elhage_etal_2022_superposition} provide a plausible mechanism by which weak distinctions become noisy under capacity pressure. This resolves the tension between the Platonic Representation Hypothesis \citep{huh_etal_2024} and the Aristotelian refinement of \citet{groeger_etal_2026}: topological convergence (shared partition) is proved, but global geometric convergence is not established.

\paragraph{Limits.}
Ecological veridicality is a claim about representational adequacy relative to a task ecology, not about honesty, calibration, or understanding in any thicker sense. A model may preserve all task-separated distinctions and still mislead at the level of surface behaviour. Some such failures may also be computational rather than representational: even when $\Delta_D(\theta)=0$, a restricted deployment inference class can leave a positive decoding gap $\Gamma_D^{\mathcal{Q}_{\mathrm{dep}}}(\theta)$, and extended inference may reduce that term without changing the frozen encoding. The finite-class learning guarantees (\cref{sec:static}) control the oracle objective $L_D^*(\theta)$, not a full optimisation theorem for realistic transformer training or a bound on $\Gamma_D^{\mathcal{Q}_{\mathrm{dep}}}(\theta)$. The geometry gap remains the main open mathematical problem; extending the mean-field analysis of \citet{wang_johnston_fusi_2025} from feedforward to attention-based architectures would be a natural route.

\paragraph{Niche construction.}
For language models, $W$ is not exogenous: model-generated text enters future training corpora, reshaping the ecology to which later generations adapt. This is a niche-construction problem \citep{laland_etal_2016}. Recent work on model collapse under synthetic-data retraining points to one concrete manifestation of that feedback: when later models train on data that no longer surprises them, performance and diversity can decay across generations \citep{gambetta_etal_2025_surplexity}. At any snapshot the framework applies; but long-run veridicality may become faithfulness to a partially model-constructed world. Recent evidence that language models can sometimes detect manipulations of their own internal states \citep{lindsey_2025_introspection} suggests a weaker, individual-level analogue of the same point: some computational states may themselves become part of the effective world the model tracks. Formalising that feedback loop would require coupling the population dynamics of \cref{sec:evolution} to a dynamics on $W$ and $\mu$, which we do not attempt here.

\medskip
Two conclusions follow. The ecological veridicality framework identifies which world-state distinctions the training ecology forces a Bayes-optimal encoding to preserve, and which it leaves free to merge. Simplicity pressure determines the order in which weak distinctions are shed. The two-ecology extension locates the gap pairs where evaluation pressure and post-training injection matter beyond the base token objective. These are specific, testable claims; they do not require geometric convergence, which the present results leave unresolved. The strongest convergence narratives therefore remain conditional: on the ecology, on penetrability, on simplicity pressure, and on whether model populations reshape the worlds to which they are supposed to become veridical. The model-organism methodology makes those conditions testable in a regime where every theoretical quantity is directly observable, and the resulting distinctions can be carried as disciplined hypotheses into the study of larger systems.

\acks{No external funding. No conflicts of interest. Thanks to Sinon, son of Autolycus.

Code and experiment scripts are available at \url{https://github.com/gvdr/llm_evo_veridicity}.}

\bibliography{refs}

\appendix
\section{Geometry of the Task Ecology}\label{app:geometry}

\subsection{Canonical Hilbert Geometry}

\begin{definition}[Task-distance kernel]
Let $N := |W|$, let $D_\sigma$ be the matrix with $(D_\sigma)_{ij} = \sigmasq(w_i, w_j)$, and let $I$ denote the $N \times N$ identity matrix. Under uniform centering
\[
J = I - (1/N)\mathbf{1}\mathbf{1}^T,
\]
define
\[
K_\mu = -\tfrac{1}{2} J D_\sigma J.
\]
(For non-uniform priors, replace $J$ by weighted centering.)
\end{definition}

\begin{proposition}[Canonical Hilbert embedding and PSD kernel]\label{prop:hilbert-kernel}
Let $\mathcal{H}_D = L^2(D_C;\allowbreak \R^{|V|})$ and define the square-root embedding $\Psi_D \colon W \to \mathcal{H}_D$ by
\[
\Psi_D(w)(c) = \sqrt{P_w(\cdot \mid c)},
\]
where the square root is taken coordinate-wise. Then for all $w_i,w_j \in W$:
\[
\sigmasq_D(w_i,w_j) = \tfrac{1}{2}\|\Psi_D(w_i) - \Psi_D(w_j)\|^2_{\mathcal{H}_D}.
\]
Consequently $D_\sigma$ is a squared Euclidean distance matrix (up to the factor $\tfrac{1}{2}$). Moreover, $K_\mu = -\tfrac{1}{2} J D_\sigma J$ is positive semidefinite, and if $\bar{\Psi}_D(w_i) = \Psi_D(w_i) - \tfrac{1}{N}\sum_j \Psi_D(w_j)$, then
\[
(K_\mu)_{ij} = \tfrac{1}{2}\langle \bar{\Psi}_D(w_i),\, \bar{\Psi}_D(w_j)\rangle_{\mathcal{H}_D}.
\]
\end{proposition}

\begin{proof}
By the definition of task distance under the training ecology,
\[
\sigmasq_D(w_i,w_j)
= \E_{c\sim D_C}\!\left[\frac12\sum_v
\bigl(\sqrt{P_{w_i}(v\mid c)}-\sqrt{P_{w_j}(v\mid c)}\bigr)^2\right].
\]
Since $\Psi_D(w)(c)=\sqrt{P_w(\cdot\mid c)}$ coordinate-wise, the right-hand side is exactly
\[
\frac12\int \sum_v
\bigl(\Psi_D(w_i)(c)_v-\Psi_D(w_j)(c)_v\bigr)^2\,dD_C(c)
= \frac12\|\Psi_D(w_i)-\Psi_D(w_j)\|_{\mathcal H_D}^2,
\]
proving the first claim.

Now center the embedded points by
\[
\bar{\Psi}_D(w_i)=\Psi_D(w_i)-\frac1N\sum_{j=1}^N \Psi_D(w_j).
\]
Subtracting the same mean vector from both points does not change pairwise differences, so
\[
\|\Psi_D(w_i)-\Psi_D(w_j)\|^2
= \|\bar{\Psi}_D(w_i)-\bar{\Psi}_D(w_j)\|^2.
\]
For any centered Euclidean point cloud, the standard double-centering identity recovers the Gram matrix from the squared distance matrix:
\[
-\frac12 JD_\sigma J
\]
is the Gram matrix of the centered representatives. Hence
\[
(K_\mu)_{ij}
= \frac12\langle \bar{\Psi}_D(w_i),\bar{\Psi}_D(w_j)\rangle_{\mathcal H_D},
\]
so $K_\mu$ is positive semidefinite.
\end{proof}

The task ecology therefore determines a canonical Hilbertian target geometry independently of any particular neural architecture. What remains non-trivial is whether a learned representation $h \colon W \to \R^d$ approximates this geometry, rather than merely preserving the induced partition.

Next we record the standard square-root-embedding fact for Hellinger geometry together with the usual double-centering construction for Euclidean distance matrices, expressed in the present notation.

\subsection{What the Framework Proves for Learned Encoders}

\begin{definition}[Ecological veridicality of a representation map]
For $h \colon W \to \R^d$, define the partition ${\sim_h}$ on $W$ by $w_i \sim_h w_j$ iff $h(w_i) = h(w_j)$, and let $p_h \colon W \to W/{\sim_h}$ be the induced encoding. Let $K_h$ denote the centered Gram matrix of the learned codes, i.e.\ the Gram matrix of $\{h(w_i)-\frac{1}{|W|}\sum_j h(w_j)\}_i$. We say that $h$ is ecologically veridical when $p_h$ merges no $\mu$-separated pair.
\end{definition}

\begin{theorem}[Topological prediction, general case]\label{thm:topology-general}
Let $h \colon W \to \R^d$ be ecologically veridical. Then:
\begin{enumerate}[label=(\alph*)]
\item $h(w_i) \neq h(w_j)$ for every $\mu$-separated pair.

\item $h(w_i) = h(w_j)$ is permitted for $\mu$-equivalent pairs. For minimum-complexity zero-excess encoders (in the sense of \thmref{thm:min-complexity}), equality on $\mu$-equivalent pairs is additionally required.

\item If $h$ realises exactly $\kmu := |W/{\sim_\mu}|$ distinct class codes, then $\rank(K_h) \le \kmu - 1$, with equality when class representatives are in affine general position.
\end{enumerate}
\end{theorem}

\begin{proof}
Part~(a) is exactly \citet[Theorem~4.1(a)]{dallariva_2026}: an ecologically veridical representation may not collapse any $\mu$-separated pair. For~(b), the same framework permits equality on $\mu$-equivalent pairs, while \thmref{thm:min-complexity}(b) adds a stronger requirement for minimum-complexity zero-excess encoders: their partition must be exactly $W/{\sim_\mu}$. For~(c), if $h$ realizes exactly $\kmu$ distinct class codes, then after centering there are still at most $\kmu$ distinct code vectors and their centered sum is zero. Their span therefore has dimension at most $\kmu-1$, so the centered Gram matrix $K_h$ has rank at most $\kmu-1$. Equality is achieved when the distinct class representatives are in affine general position.
\end{proof}

\begin{remark}[What this does NOT constrain]
\thmref{thm:topology-general} constrains only the partition structure induced by $h$ and the resulting rank bound on $K_h$. It does NOT constrain relative magnitudes of non-zero distances, eigenvectors, or overall scale.
\end{remark}

\subsection{Exact Metric Recovery in the Gaussian-Linear Case}

\begin{remark}[Scope of Appendix A.3]
In this section, we use a simplified Gaussian-linear model rather than the autoregressive setting of the main paper. Tasks are scalar-valued ($f(w_i) = c^T \varphi_i$), not distribution-valued ($f_c(w) = P_w(\cdot\mid c)$). The results here illustrate when geometric alignment (beyond the topological alignment proved in the main text) is achievable, and identify the restrictive conditions under which it holds.
\end{remark}

\begin{theorem}[Geometric alignment, Gaussian-linear case]\label{thm:geometry-gaussian}
Consider Gaussian-linear tasks $f(w_i) = c^T \varphi_i$ with $c \sim N(0, \Sigma_c)$, a linear encoder $h(w_i) = A\varphi_i$, and the task-relevant subspace
$V = \mathrm{span}\{\varphi_i - \varphi_j : \sigmasq(w_i,w_j)>0\}$.
Write $P_V$ for the orthogonal projector onto~$V$ and $\Delta = \{\varphi_i - \varphi_j : \sigmasq(w_i,w_j) > 0\}$ for the set of separated difference vectors. Assume readouts attain Bayes-optimal prediction on $h$-cells. Then:
\begin{enumerate}[label=(\alph*)]
\item Zero risk requires $A(\varphi_i - \varphi_j) \neq 0$ for every separated pair, i.e.\ $\Delta \cap \ker(A) = \emptyset$. A sufficient (but not necessary) condition is $\ker(A) \cap V = \{0\}$.

\item Under the sufficient condition $\ker(A) \cap V = \{0\}$, the minimum feasible rank is $\dim(V)$. Without it, lower ranks may suffice if the finitely many vectors in $\Delta$ avoid $\ker(A)$.

\item In the canonical projector gauge $A = P_V$ (which satisfies the sufficient condition):
$\|h(w_i) - h(w_j)\|^2 = \|P_V(\varphi_i - \varphi_j)\|^2$.

If $\Sigma_c = \sigma_c^2 P_V$ (isotropic on $V$):
$\|h(w_i) - h(w_j)\|^2 = \sigmasq(w_i, w_j) / \sigma_c^2$, i.e.\ exact proportionality.

If $\Sigma_c$ is anisotropic: exact proportionality is no longer guaranteed and generically fails. The encoder projects onto $V$ uniformly, while $\sigmasq$ weights directions by $\Sigma_c$.
\end{enumerate}
\end{theorem}

\begin{proof}
By \citet[Theorem~4.1(a)]{dallariva_2026}, zero risk is equivalent to merging only $\mu$-equivalent states. In the linear setting, $h(w_i) = h(w_j)$ iff $A(\varphi_i - \varphi_j) = 0$, so zero risk requires $A(\varphi_i - \varphi_j) \neq 0$ for every separated pair, proving~(a). For~(b), $\ker(A) \cap V = \{0\}$ implies $A$ is injective on $V$, so $\rank(A) \ge \dim(V)$, with equality achievable. For~(c), pick the canonical representative $A = P_V$. For isotropic $\Sigma_c$, $\sigmasq(w_i,w_j) = \sigma_c^2\|P_V(\varphi_i - \varphi_j)\|^2$, giving proportionality. For anisotropic $\Sigma_c$, write the spectral decomposition of $\Sigma_c$ on $V$ as $\Sigma_c|_V = \sum_k \lambda_k u_k u_k^T$, where $\{u_k\}$ is an orthonormal basis of $V$ and $\lambda_k > 0$ are the corresponding directional variances. Then $\sigmasq = \sum_k \lambda_k[(\varphi_i-\varphi_j)^T u_k]^2$ while $\|P_V(\varphi_i-\varphi_j)\|^2 = \sum_k[(\varphi_i-\varphi_j)^T u_k]^2$; proportional iff all $\lambda_k$ equal.
\end{proof}

\subsection{Neighborhood Stability and the Open Problem}

The main topological convergence result appears in the body of the paper as \corref{cor:topological-convergence}. Here we record only the neighborhood-stability lemma and the remaining open geometric question.

\begin{proposition}[Neighborhood recovery from metric approximation]\label{prop:knn-stability}
Let $d$ be a target metric on $W$ and $\hat d$ a learned metric. Fix $k$ and, for each $w_i$, let $r_k(i)$ and $r_{k+1}(i)$ denote the distances from $w_i$ to its $k$-th and $(k+1)$-st nearest neighbors under $d$. Assume the $k$-neighborhood margin
\[
\gamma_k := \min_i \bigl(r_{k+1}(i) - r_k(i)\bigr)
\]
is strictly positive. If
\[
\sup_{i \neq j} |\hat d(w_i,w_j) - d(w_i,w_j)| < \gamma_k/2,
\]
then the directed $k$-nearest-neighbor graph induced by $\hat d$ is exactly the same as the one induced by $d$.
\end{proposition}

\begin{proof}
Fix $w_i$. Every true $k$-nearest neighbor $w_j$ of $w_i$ satisfies $d(w_i,w_j) \le r_k(i)$, hence $\hat d(w_i,w_j) < r_k(i) + \gamma_k/2$. Every point $w_\ell$ outside the true $k$-neighborhood satisfies $d(w_i,w_\ell) \ge r_{k+1}(i)$, hence $\hat d(w_i,w_\ell) > r_{k+1}(i) - \gamma_k/2$. Because $r_{k+1}(i)-\gamma_k/2 \ge r_k(i)+\gamma_k/2$, no outsider can cross into the top-$k$ set under $\hat d$, and no true member can be pushed out. Since this holds for every~$i$, the directed $k$-NN graphs coincide.
\end{proof}

\begin{remark}[Status]
\propref{prop:knn-stability} is a standard margin-based perturbation lemma for nearest-neighbor graphs, included here for completeness rather than as a novel result.
\end{remark}

\begin{remark}[Open problem]
\propref{prop:hilbert-kernel} identifies the target geometry induced by the ecology, and \thmref{thm:geometry-gaussian} proves exact recovery only in a restrictive Gaussian-linear regime. For deep non-linear learners, existing feature-learning theory suggests partial alignment with the leading eigendirections of $K_\mu$ \citep{ba_etal_2022,atanasov_etal_2022,damian_etal_2022}, but does not establish full proportional recovery of pairwise distances. \propref{prop:knn-stability} shows what would be sufficient for Aristotelian local-neighborhood recovery, but the required metric-approximation theorem is only proved here in the Gaussian-linear case. General geometric convergence is therefore unresolved by the present results.
\end{remark}

\begin{interpretation}
Our framework supplies a canonical ecology kernel $K_\mu$ and proves that minimum-complexity zero-excess models agree on the induced partition. Exact recovery of $K_\mu$ by learned distances is proved only in the Gaussian-linear isotropic case. The prediction absent from \citet{huh_etal_2024} is \emph{failure}: when deployment probes distinctions weakly constrained by training, models may diverge rather than converge.
\end{interpretation}

\section{Deployment Decoder Classes}\label{app:decoder-gap}

The main text introduces the deployment decoding gap
\[
\Gamma_D^{\mathcal{Q}_{\mathrm{dep}}}(\theta)
= L_D^{\mathcal{Q}_{\mathrm{dep}}}(\theta)-L_D^*(\theta),
\]
which isolates the difference between the Bayes-optimal decoder for a fixed induced encoding and the best decoder available under a restricted deployment inference regime. Here we generalize \defref{def:decoder-gap} from induced encodings $p_{\theta,D}$ to arbitrary encodings $p$, and record only the structural facts needed in the body of the paper.

\begin{definition}[Decoder-class loss for a fixed encoding]
Let $p \colon W \to X$ be any encoding and let $\mathcal{Q}$ be a nonempty class of decoders
\[
q \colon X \times V^* \to \Delta(V).
\]
Define the best loss achievable within $\mathcal{Q}$ by
\[
L_D^{\mathcal{Q}}(p)
:= \inf_{q \in \mathcal{Q}} L_D(p,q),
\]
and the corresponding decoder-class gap by
\[
\Gamma_D^{\mathcal{Q}}(p)
:= L_D^{\mathcal{Q}}(p)-L_D^*(p).
\]
The infimum need not be attained in general; when it is attained, any minimizer is a best $\mathcal{Q}$-decoder for~$p$.
\end{definition}

\begin{proposition}[Monotonicity under decoder-class expansion]\label{prop:decoder-gap-monotone}
Let $\mathcal{Q}_1 \subseteq \mathcal{Q}_2$ be two nonempty decoder classes for the same encoding~$p$. Then
\[
L_D^{\mathcal{Q}_2}(p) \le L_D^{\mathcal{Q}_1}(p)
\qquad\text{and}\qquad
\Gamma_D^{\mathcal{Q}_2}(p) \le \Gamma_D^{\mathcal{Q}_1}(p).
\]
\end{proposition}

\begin{proof}
Because $\mathcal{Q}_1 \subseteq \mathcal{Q}_2$, taking the infimum over the larger class cannot increase the value:
\[
\inf_{q\in\mathcal{Q}_2} L_D(p,q)
\le
\inf_{q\in\mathcal{Q}_1} L_D(p,q).
\]
Subtracting the common baseline $L_D^*(p)$ gives the same inequality for the decoder-class gaps.
\end{proof}

\begin{corollary}[Induced-encoding form]
If $\mathcal{Q}_1 \subseteq \mathcal{Q}_2$ are nonempty deployment decoder classes, then for every $\theta \in \Theta$:
\[
L_D^{\mathcal{Q}_2}(\theta) \le L_D^{\mathcal{Q}_1}(\theta)
\qquad\text{and}\qquad
\Gamma_D^{\mathcal{Q}_2}(\theta) \le \Gamma_D^{\mathcal{Q}_1}(\theta).
\]
\end{corollary}

\begin{proof}
Apply \propref{prop:decoder-gap-monotone} to the induced encoding $p_{\theta,D}$.
\end{proof}

\begin{remark}[Interpretation]
Single-pass prompting, chain-of-thought prompting, scratchpads, and tool-augmented inference can be idealized as different deployment decoder classes for the same frozen encoding. The proposition above therefore supports only the monotonic claim used in the main text: if one inference regime genuinely enlarges the admissible decoder family relative to another, then the best achievable decoding gap cannot increase. Establishing concrete inclusion relations among realistic transformer prompting strategies, or bounding the resulting gaps for specific architectures, is a separate circuit-complexity problem that we do not attempt here.
\end{remark}

\section{Supplementary Consequences}

\subsection{Generalist versus Specialist}\label{sec:genspec}

The generalist-specialist comparison gives a supplementary consequence of the same excess decomposition: broad ecologies favor representations that preserve all distinctions jointly required across tasks, while specialists can be optimal only relative to narrower sub-ecologies.

\begin{theorem}[Generalist advantage]\label{thm:gen-vs-spec}
For each task $t \in \{1,\ldots,T\}$, let $D^{(t)}$ be the corresponding data distribution, $\mu_t$ its induced task ecology, and $D_C^{(t)}$ its context marginal.
Define the excess Bayes-optimal token loss
\[
\Delta_t(\theta) := L_{D^{(t)}}^*(\theta) - H_{D^{(t)}}(Y \mid C,W).
\]
For each $t$, interpret $L_{D^{(t)}}^*$ and $H_{D^{(t)}}$ under the joint law induced by $\pi$, $D_C^{(t)}$, and the conditional token distributions $P_w(\cdot\mid c)$. Let $D := (1/T)\sum_t D^{(t)}$ be the uniform task mixture, with induced ecology $\mu_D$ and context marginal $D_C := (1/T)\sum_t D_C^{(t)}$. Then:
\begin{enumerate}[label=(\alph*)]
\item If $\Delta_D(\theta_G)=0$, then $\Delta_t(\theta_G)=0$ for all $t$: the generalist matches every specialist on each constituent task.

\item For any specialist $\theta_t$ and task $s \neq t$: if $\theta_t$ merges a pair $(w_1,w_2)$ with $\sigmasq_{\mu_s}(w_1,w_2)>0$, then
\[
\Delta_s(\theta_t) > 0.
\]
More explicitly, if $x=p_{\theta_t,D^{(s)}}(w_1)=p_{\theta_t,D^{(s)}}(w_2)$ and $\lambda=\pi(w_1)/(\pi(w_1)+\pi(w_2))$, then
\[
\Delta_s(\theta_t)
\ge (\pi(w_1)+\pi(w_2))\,
   \E_{c\sim D_C^{(s)}}
   \!\left[\JS_{\lambda}
   \bigl(P_{w_1}(\cdot\mid c),\,P_{w_2}(\cdot\mid c)\bigr)\right]
> 0.
\]

\item Hence, on any deployment distribution that gives positive weight to at least one such missed pair, a zero-excess generalist strictly outperforms that specialist in Bayes-optimal next-token loss.
\end{enumerate}
\end{theorem}

\begin{proof}
(a)~If $\Delta_D(\theta_G)=0$, \thmref{thm:ce-decomposition}(c) implies that $p_{\theta_G,D}$ merges only pairs that are $D_C$-almost-everywhere equivalent under the mixture ecology. Since $D_C=(1/T)\sum_t D_C^{(t)}$, this implies $D_C^{(t)}$-almost-everywhere equivalence for every~$t$. Hence $\Delta_t(\theta_G)=0$ for all~$t$.

(b)~Let $x$ be the merged cell containing $w_1$ and $w_2$ under task~$s$. By \thmref{thm:ce-decomposition}(b), the contribution of cell $x$ to $\Delta_s(\theta_t)$ is
\[
\pi_x\, \E_{c\sim D_C^{(s)}}
\bigl[\JS_{\alpha_x}
(\{P_w(\cdot\mid c)\}_{w\in C_x})\bigr].
\]
Grouping all states in $C_x \setminus \{w_1,w_2\}$ into a residual component gives the exact decomposition
\[
\JS_{\alpha_x}(\{P_w\}_{w\in C_x})
= \JS_{(\beta,1-\beta)}(M_{12},M_{\mathrm{rest}})
  + \beta\,\JS_{\lambda}(P_{w_1},P_{w_2})
  + (1-\beta)\,\JS_{\mathrm{rest}},
\]
where $\beta=\alpha_x(w_1)+\alpha_x(w_2)$, $M_{12}$ is the $(\lambda,1-\lambda)$-mixture of $P_{w_1},P_{w_2}$, $M_{\mathrm{rest}}$ is the mixture of the remaining cell distributions, and $\JS_{\mathrm{rest}}\ge 0$ is their internal weighted Jensen--Shannon divergence. Hence the cell contribution is at least
\[
\pi_x\beta\,\E_{c\sim D_C^{(s)}}\bigl[\JS_{\lambda}(P_{w_1}(\cdot\mid c),P_{w_2}(\cdot\mid c))\bigr],
\]
which is the displayed bound because $\pi_x\beta=\pi(w_1)+\pi(w_2)$. Because $\sigmasq_{\mu_s}(w_1,w_2)>0$, the two next-token laws differ on a set of positive $D_C^{(s)}$-measure, so the two-state Jensen--Shannon term is positive on a set of positive measure and therefore has strictly positive expectation.

(c)~From~(a) and~(b).
\end{proof}

\section{Selection on Recipe Traits}\label{app:recipe}

The two-ecology framework of \cref{sec:two-ecology} identifies post-training as an ecology-injection mechanism. The following results characterise how lineage selection acts on the strength of that injection.

\begin{proposition}[Selection on recipe traits]\label{prop:recipe-trait}
Let $\mathcal{R}$ be a finite recipe space with a heritable scalar trait $\alpha(r)\in[0,1]$ for each $r\in\mathcal{R}$, interpreted as the strength of ecology injection.
Consider one selection stage in a Wright--Fisher population over recipes $r_1,\ldots,r_K$ with frequencies~$x(r)$ and expected fitness~$f(r)$.
Define the population mean trait $\bar\alpha := \sum_r x(r)\,\alpha(r)$ and let $\bar\alpha_{\mathrm{eval}}$ denote the mean trait in the selected-parent population.
Then:
\begin{enumerate}[label=(\alph*)]
\item \textbf{Exact identity.}
$\bar\alpha_{\mathrm{eval}} - \bar\alpha = \Cov_x(f,\alpha)\,/\,\bar f$.

\item \textbf{Sufficient condition.} Write each recipe as $r=(\alpha,\zeta)$, where~$\zeta$ collects all other coordinates. If for every fixed~$\zeta$ the map $\alpha\mapsto\Delta_{\mathrm{eval}}(\alpha,\zeta)$ is nonincreasing, and if fitness has the form $f(r)=g(\Delta_{\mathrm{eval}}(r))$ for a strictly decreasing~$g$, then $\Cov_x(f,\alpha)\ge 0$ and therefore $\bar\alpha_{\mathrm{eval}}\ge\bar\alpha$.

\item \textbf{Strict increase.} If, in addition, there is positive recipe mass on a set of~$\zeta$ values for which $\alpha\mapsto\Delta_{\mathrm{eval}}(\alpha,\zeta)$ is strictly decreasing on a set of positive conditional mass, then $\Cov_x(f,\alpha)>0$ and $\bar\alpha_{\mathrm{eval}}>\bar\alpha$.
\end{enumerate}
\end{proposition}

\begin{proof}
The selected-parent distribution is $x_{\mathrm{sel}}(r) = x(r)\,f(r)/\bar f$, so
\[
\bar\alpha_{\mathrm{eval}} = \frac{1}{\bar f}\sum_r x(r)\,f(r)\,\alpha(r),
\]
giving~(a). For~(b), under the stated monotonicity the conditional mean fitness $m(a):=\E_x[f\mid\alpha=a]$ is nondecreasing in~$a$. Using an independent copy~$A'$ of~$A$,
\[
2\,\Cov_x(m(A),A)=\E[(m(A)-m(A'))(A-A')]\ge 0.
\]
Since $\Cov_x(f,\alpha)=\Cov_x(m(\alpha),\alpha)$ by the law of total covariance, the claim follows. For~(c), strict decrease on positive mass gives $P((m(A)-m(A'))(A-A')>0)>0$, hence strict positivity.
\end{proof}

The monotonicity condition in~(b) is substantive. It can fail if the injected task family is badly aligned with the evaluation ecology, for example through reward hacking, capability degradation, or post-training that improves a proxy while worsening the actual deployment target.

\begin{lemma}[Selection-stage directional gap closing]\label{lem:gap-closing}
Fix a gap pair $(w_1,w_2)\in G_\varepsilon$ and define the recipe-level token-ecology separation score $s(r):=\sigmasq_{\mathrm{tok}}(r;\,w_1,w_2)$. Assume $s(r)$ depends on recipes only through the scalar trait~$\alpha(r)$, and write $s(r)=h(\alpha(r))$ for some nondecreasing~$h$. If the selected-parent trait distribution first-order stochastically dominates the pre-selection distribution, then $\E_{\mathrm{eval}}[h(\alpha)]\ge\E[h(\alpha)]$.
\end{lemma}

\begin{proof}
By the standard monotone-comparison property of first-order stochastic dominance applied to the nondecreasing function~$h$.
\end{proof}

The assumption $s(r)=h(\alpha(r))$ collapses all other recipe coordinates into a single scalar trait and assumes monotone dependence on that trait alone. The lemma is therefore an idealised strengthening that guides the design of controlled synthetic experiments, rather than a claim about realistic recipe spaces.

\section{Off-Ecology Error Bounds}\label{app:off-ecology}

The following propositions provide quantitative bounds for the failure predictions stated in \cref{sec:conclusion}. We prove them only for the next-token log-loss ecology. Extending them to generalized ecologies would require additional assumptions on the task losses; we do not use that extension in the present manuscript.

Let $\mu_1$ be the ecology under which the encoding was optimized and let $\mu_2$ be a probe ecology that refines~$\mu_1$. Let $D_C^{(1)}$ and $D_C^{(2)}$ denote context marginals inducing $\mu_1$ and $\mu_2$, respectively. For $i\in\{1,2\}$, interpret $L_{\mu_i}^*$ and $H_{\mu_i}$ under the joint law induced by $\pi$, $D_C^{(i)}$, and the conditional token distributions $P_w(\cdot\mid c)$.

\begin{proposition}[Off-ecology excess bound]\label{prop:off-ecology-error}
Let $p$ be a minimum-complexity zero-excess encoding for~$\mu_1$. If $\sigmasq_{\mu_1}(w_1,w_2)=0$ and $\sigmasq_{\mu_2}(w_1,w_2)>0$, then $p(w_1)=p(w_2)$ and
\[
L_{\mu_2}^*(p) - H_{\mu_2}(Y \mid C,W)
\ge (\pi(w_1)+\pi(w_2))\,
   \E_{c\sim D_C^{(2)}}
   \!\left[\JS_{\lambda}
   \bigl(P_{w_1}(\cdot\mid c),\,P_{w_2}(\cdot\mid c)\bigr)\right]
> 0,
\]
where $\lambda=\pi(w_1)/(\pi(w_1)+\pi(w_2))$.
\end{proposition}

\begin{proof}
By \thmref{thm:min-complexity}, a minimum-complexity zero-excess encoding for~$\mu_1$ has partition $W/{\sim_{\mu_1}}$. Since $\sigmasq_{\mu_1}(w_1,w_2)=0$, we have $w_1\sim_{\mu_1}w_2$, hence $p(w_1)=p(w_2)$. Let $x$ be that merged cell. By \thmref{thm:ce-decomposition}(b), the contribution of cell~$x$ to the excess under~$\mu_2$ is
\[
\pi_x\, \E_{c\sim D_C^{(2)}}
\bigl[\JS_{\alpha_x}
(\{P_w(\cdot\mid c)\}_{w\in C_x})\bigr].
\]
Group the states in $C_x$ into the pair $\{w_1,w_2\}$ and the residual set $C_x\setminus\{w_1,w_2\}$. The same hierarchical weighted Jensen--Shannon decomposition used in \cref{sec:genspec} gives
\[
\JS_{\alpha_x}(\{P_w\}_{w\in C_x})
\ge \beta\,\JS_{\lambda}(P_{w_1},P_{w_2}),
\]
where $\beta=\alpha_x(w_1)+\alpha_x(w_2)$ and $\lambda=\alpha_x(w_1)/\beta=\pi(w_1)/(\pi(w_1)+\pi(w_2))$. Multiplying by $\pi_x$ yields the displayed lower bound because $\pi_x\beta=\pi(w_1)+\pi(w_2)$. Since $\sigmasq_{\mu_2}(w_1,w_2)>0$, the two next-token laws differ on a set of positive $D_C^{(2)}$-measure, so the two-state Jensen--Shannon term has strictly positive expectation.
\end{proof}

\begin{proposition}[Off-ecology non-identifiability]\label{prop:off-ecology-nonident}
Under the same setup, if there exists a context set~$A$ with $D_C^{(1)}(A)=0$, $D_C^{(2)}(A)>0$, and $P_{w_1}(\cdot\mid c)\neq P_{w_2}(\cdot\mid c)$ for $c\in A$, then there exist two decoders $q^{(1)},q^{(2)}$ that attain the same optimal loss under~$\mu_1$ but disagree on~$A$. The training objective does not identify a unique off-ecology extension.
\end{proposition}

\begin{proof}
Let $x:=p(w_1)=p(w_2)$; by assumption, the probe ecology distinguishes two states that the optimized encoding leaves merged. Set both decoders equal to the Bayes-optimal decoder for~$p$ under~$\mu_1$ outside~$A$. On~$A$, define
\[
q^{(1)}(x,c):=P_{w_1}(\cdot\mid c),
\qquad
q^{(2)}(x,c):=P_{w_2}(\cdot\mid c),
\]
and leave all other code cells unchanged. Since $D_C^{(1)}(A)=0$, these modifications affect a set of zero $\mu_1$-measure, so both decoders attain the same optimal $\mu_1$-loss. Since $D_C^{(2)}(A)>0$ and the laws differ on~$A$, the two off-ecology extensions disagree on a set of positive probe measure.
\end{proof}

\section{Corpus Sources}\label{app:corpora}

We normalized all corpora to ASCII-range characters. We transliterated Unicode accented characters, removed markup, headers, and metadata, and split each text into fixed-length character chunks for tokenization.

\paragraph{Alice's Adventures in Wonderland.}
Five languages: English, French (trans.\ Henri Bu{\'e}), German (trans.\ Antonie Zimmermann), Italian (trans.\ T.\ Pietroc{\`o}la-Rossetti), Finnish (trans.\ Anni Swan). Digital texts from Project Gutenberg ebooks \#11, \#55456, \#19778, \#28371, \#46569.

\paragraph{Dante's Commedia.}
Seven languages: Italian, English, German, Finnish, Spanish, French, Portuguese. Digital texts from Project Gutenberg ebooks \#1000, \#1004, \#8085, \#12546, \#57303, \#22768/\#22769, and Portuguese text from pt.Wikisource.

\paragraph{Communist Manifesto.}
Ten languages: English, German, Spanish, French, Italian, Portuguese, Polish, Czech, Dutch, Finnish. Digital texts from the Marxists Internet Archive (\url{https://www.marxists.org/}).

\paragraph{Voynich manuscript.}
EVA transliteration in IVTFF format from Rene Zandbergen's digital archive (\url{https://www.voynich.nu/transcr.html}), using Takeshi Takahashi's complete transcription. We retained only lowercase Latin-alphabet characters, i.e.\ the EVA encoding of Voynich glyphs.

\paragraph{Practical Common Lisp.}
Source code from Peter Seibel's \emph{Practical Common Lisp}, normalized to lowercase letters, parentheses, and spaces. We use it for the bracket-balance and code-validation experiments (\cref{sec:two-ecology}). We verified balanced chunks for proper bracket nesting and generated unbalanced chunks by randomly permuting bracket characters at the same positions.

\end{document}